\documentclass[aps,prd,showpacs,amssymb,superscriptaddress,nofootinbib,twocolumn,notitlepage]{revtex4-1}
\usepackage{amsmath,amsfonts,amsthm,amssymb,verbatim,color}
\usepackage{siunitx}
\usepackage{graphicx}
\usepackage{braket}
\usepackage{float}
\usepackage{epstopdf}
\usepackage{hyperref}
\usepackage{subcaption}
\usepackage[top=0.6in, bottom=0.7in, left=0.6in, right=0.7in]{geometry}
\usepackage[titletoc,toc,page]{appendix}
\usepackage{times}
\newtheorem{theorem}{Theorem}

\newtheorem{definition}{Definition}
\newtheorem{corollary}{Corollary}
\newcommand{\tr}{\textup{tr}}
\newcommand{\id}{\mathbb{I}}
\newenvironment{myproof}[1]{{\emph{Proof of {#1}.} }}{\hfill $\Box$ \\}

\theoremstyle{definition}

\theoremstyle{Theorem}
\newtheorem{Lemma}{Lemma}

\theoremstyle{Theorem}

\theoremstyle{Theorem}

\theoremstyle{definition}

\newcommand{\proj}[1]{\ensuremath{| #1\rangle\!\langle #1 |} }
\newcommand{\steepest}{\rho_{\rm st}^\varepsilon}
\newcommand{\steep}{\rho_{\rm steep}^\varepsilon}
\newcommand{\epssteep}[1]{\rho_{\rm steep}^{#1}}
\newcommand{\flattest}{\rho_{\rm fl}^\varepsilon}

\newcommand{\bdepsball}{\mathcal{B}^\varepsilon_{\rm D}}
\newcommand{\epsball}{\mathcal{B}^\varepsilon}\renewcommand{\epsilon}{\varepsilon}
\newcommand{\sfe}{\hat F_\alpha^\varepsilon}
\newcommand{\gfe}{generalized free energies}
\newcommand{\eig}{\mathrm{eig}}
\newcommand{\sgn}{\textup{sgn}}

\begin{document}
	
	\title{Smoothed generalized free energies for thermodynamics}
	\author{Remco van der Meer}
	\affiliation{QuTech, Delft University of Technology, Lorentzweg 1, 2611 CJ Delft, Netherlands}
	\author{Nelly Ng}
	\affiliation{QuTech, Delft University of Technology, Lorentzweg 1, 2611 CJ Delft, Netherlands}
	\affiliation{Centre for Quantum Technologies, National University of Singapore, 117543 Singapore}
	\author{Stephanie Wehner}
	\affiliation{QuTech, Delft University of Technology, Lorentzweg 1, 2611 CJ Delft, Netherlands}
	
	\begin{abstract}
	In the study of thermodynamics for nanoscale quantum systems, a family of quantities known as generalized free energies have been derived as necessary and sufficient conditions that govern state transitions. These free energies become important especially in the regime where the system of interest consists of only a few (quantum) particles. In this work, we introduce a new family of smoothed \gfe, by constructing explicit smoothing procedures that maximize/minimize the free energy over an $ \varepsilon$-ball of quantum states. In contrast to previously known smoothed free energies, these quantities now allow us to make an operational statement for \emph{approximate} thermodynamic state transitions.
	We show that these newly defined smoothed quantities converge to the standard free energy in the thermodynamic limit.
	\end{abstract}
	\maketitle
\section{Introduction}
The resource theory approach in quantum thermodynamics \cite{BMORS13,HO03,HO13,2ndlaw} provides a fundamental framework for understanding non-equilibrium state transitions $ \rho_S\rightarrow \sigma_S $, enabled by interactions with a larger thermal bath while conserving total energy.
Specifically, a very general class of operations studied recently are referred to as catalytic thermal operations (CTO)~\cite{2ndlaw}. Such operations take the form
\begin{align}\label{eq:CTO}
U\left(\rho_{S} \otimes \rho_C \otimes \tau_{B_\beta}\right)U^{\dagger} = \sigma_{SCB}\ ,
\end{align}
where 
$\tau_{B_\beta} = \frac{e^{-\beta  H_B}}{\tr(e^{-\beta H_B})}$ 
is the thermal state of the surrounding bath (B) with Hamiltonian $H_B$ at a fixed inverse temperature $\beta$. The system (S) has a Hamiltonian $H_S$, and is initially in the state $\rho_S$. A catalyst (C) with Hamiltonian $H_C$ is allowed, where $\rho_C$ is the initial state of the catalyst, while $ U $ is a unitary operator such that $[U,H_{\rm total}]=0$, where $H_{\rm total}=H_S+H_C+H_B$.
The latter condition simply implies that $U$ conserves total energy. Due to its generic feature, CTOs have been applied to study various scenarios in thermodynamics, such as quantum heat engines~\cite{woods2015maximum,surpassCarnot,chubb2017beyond,mueller2017correlating}, and this can be done by modeling additional systems as part of the system/catalyst if required. We say a particular transition 
\begin{align}
\rho_S \xrightarrow[\rm CTO]~\sigma_S
\end{align}
is possible, if there exist $H_B$, $H_C$, $\rho_C$ and $U$ such that Eq.~\eqref{eq:CTO} is satisfied in the regime
of exact catalysis, i.e., $\rho_C = \sigma_C = \tr_{B}\left(\sigma_{SCB}\right) = \sigma_S\otimes\rho_C$. In other words, after tracing out the surrounding heat bath, the catalyst returns to its initial state and is also uncorrelated with the system $ S $.

Phrased in this way, it may seem like a daunting task to decide whether a specific transition is possible via CTO. 
Fortunately, there exist a set of simple conditions~\cite{2ndlaw} in terms of a family of generalized free energies $ F_\alpha $, which are necessary conditions for such a state transition to happen. In other words, if $\rho_S\xrightarrow[\rm CTO]~\sigma_S$, then for all $\alpha \in\mathbb{R}$,
\begin{equation}\label{eq:exactgenfree}
F_\alpha( \rho_S,\tau_{S_\beta}) \geq F_\alpha( \sigma_S,\tau_{S_\beta}),
\end{equation}
where 
$\tau_{S_\beta} = \frac{e^{-\beta H_S}}{\tr(e^{-\beta H_S})}$ 
is the thermal state at inverse temperature $\beta$ of the surrounding bath. The usual Helmholtz free energy corresponds to the case of $\alpha \rightarrow 1$. 
Interestingly, these conditions become sufficient if the states $\rho_S$ and $\sigma_S$ are already block-diagonal in the ordered energy eigenbasis \footnote{Throughout the manuscript we refer to such states as block-diagonal states.}; or in other words, $\rho_S$ and $\sigma_S$ commute with $ H_S $. Moreover, in most cases, only the \gfe~with $ \alpha\geq 0 $ matter, since the $ \alpha <0 $ conditions may be fulfilled by borrowing a qubit ancilla and returning it extremely close to its original state \cite{2ndlaw}. These quantities signify how finite-sized quantum systems differ thermodynamically from classical macroscopic systems. Intuitively, 
these quantities also tell us that more moments of the energy distribution are indispensable in determining thermodynamical properties of a system, when we are outside a regime where the law of large numbers applies.

\begin{figure}[h!]
	\includegraphics[scale=0.73]{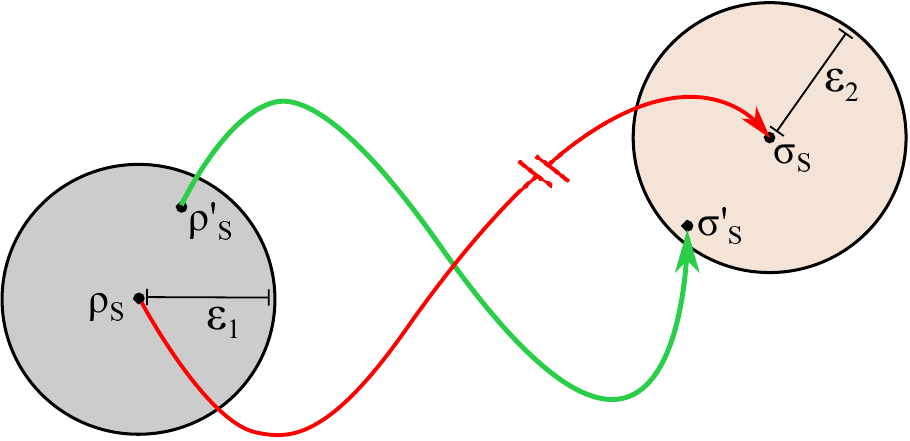}
	\caption{\small When an exact transition $ \rho_S\xrightarrow[\rm CTO]~\sigma_S $ is not possible (denoted by a disconnected red arrow), it might still be true that a state $ \rho_S' $ $ \varepsilon_1 $-close to $ \rho_S $ can be transformed to $ \sigma_S'$  $\varepsilon_2$-close to $\sigma_S $. What are the conditions governing such approximate transitions?}\label{fig:sketch}
\end{figure}

While most literature on thermodynamic resource theories is concerned with exact state transformations \cite{HO13,2ndlaw,lostaglio2015description,gour2015resource,renes2014work,faist2015gibbs,cwiklinski2014towards}, in realistic implementations, we may be satisfied as long as the transition is approximately achieved. For example, in experimental setups, initial states are prepared (and processes are implemented) always up to some high but finite accuracy \cite{PhysRevLett.113.140601,an2015experimental}, resulting in the achievement of the final state (or work distribution) up to small but non-zero errors. This has also been studied theoretically in the context of probabilistic thermal operations \cite{alhambra2016fluctuating}, using a catalyst and returning it approximately \cite{2ndlaw}, and in work extraction protocols when heat/entropy is inevitably produced alongside \cite{aaberg2013truly,woods2015maximum}.
Here, we ask whether one can identify conditions for approximate state transitions on the system $ S $ to occur, where by ``approximate" we refer to a situation in which the error $\varepsilon $ in terms of trace distance between an ideal state $ \rho $ versus the real state $ \rho' $ is small, which we also write as $ \rho'\approx_{\varepsilon}\rho $. As the trace distance quantifies how well two states can be distinguished~\cite{nielsen2002quantum}, approximate thus means that the two states are nearly indistinguishable (up to error $ \varepsilon $) by any physical process.

In this work, we make progress towards answering the question of approximate state transitions, by introducing a new family of smooth generalized free energies, $ \hat F_\alpha^\varepsilon (\rho_S,\tau_{S_\beta}) $ for any block-diagonal state $ \rho_S $. These smooth \gfe~jointly provide sufficient conditions for approximate state transitions. More precisely, if for some $ 0<\varepsilon_1,\varepsilon_2 <1 $,
\begin{equation}\label{key}
\hat F_\alpha^{\varepsilon_1} (\rho_S,\tau_{S_\beta}) \geq \hat F_\alpha^{\varepsilon_2} (\sigma_S,\tau_{S_\beta})\qquad \forall \alpha\geq 0,
\end{equation}
then we know that there exists a CTO that can take an initial state $ \epssteep{\varepsilon_1}$ $ \varepsilon_1 $-close to $ \rho_S $, to a final state $ \sigma_{\rm fl}^{\varepsilon_2} $ which is $\varepsilon_2 $-close to $ \sigma_S $. The exact form of these states $ \epssteep{\varepsilon_1}$ and $\sigma_{\rm fl}^{\varepsilon_2}$ may be explicitly determined. Moreover, the thermal operation that brings $ \epssteep{\varepsilon_1}\xrightarrow[\rm TO]~\sigma_{\rm fl}^{\varepsilon_2} $, when acted on $ \rho_S $, will also produce a final state (see Appendix \ref{subsec:approxTO})
\begin{equation}\label{key}
\rho_S \xrightarrow[\rm TO]~ \sigma_S'\approx_{\varepsilon_1+\varepsilon_2} \sigma_S.
\end{equation}
We also proved that for all $ \alpha\geq 0$, when one takes $ n $ identical and independently distributed (i.i.d.) copies, then in the limit $ n\rightarrow\infty $, and $ \varepsilon\rightarrow 0 $, the normalized quantities $ \sfe $ converge to $ F_1 $, which is the standard Helmholtz free energy known in thermodynamics. This establishes with full rigour that approximate state transitions approaching the thermodynamic limit become determined solely by the Helmholtz free energy.



\section{New divergences}
In this section, we present the form of our newly defined smooth generalized free energies. To do so, let us first recall that the exact generalized free energies are given by
\begin{equation}\label{eq:alphaFE}
F_\alpha (\rho_S,\tau_{S_\beta}):= \beta^{-1}\cdot \left[-\ln Z_\beta + D_\alpha (\rho_S\|\tau_{S_\beta})\right],
\end{equation}
where $ Z_\beta = \tr (e^{- \beta H_S})$ is the partition function, and $ D_\alpha (\rho_S\|\tau_{S_\beta}) $
are quantum R{\'e}nyi divergences defined in \cite{muller2013quantum}\footnote{The values of $ D_\alpha $ at points $ \alpha=1, \pm \infty $ are determined by the limits $ \alpha\rightarrow 1,\pm \infty $ respectively, and therefore $ D_\alpha $ is continuous in $ \alpha\in\mathbb{R} $. In Ref.~\cite{muller2013quantum}, these divergences were defined only for $ \alpha \geq 0$, however one may extend these divergences for $ \alpha <0 $, with the function $ \sgn(\alpha) $ as shown in Ref.~\cite{2ndlaw}.}. If we consider states $ \rho_S $ block-diagonal with respect to $ H_S $, then such states commute with $ \tau_{S_\beta} $. Therefore, by denoting the ordered eigenvalues of $ \rho_S,\tau_{S_\beta} $ as $ \lbrace p_i \rbrace_{i}$ and $\lbrace \tau_i \rbrace_{i} $ respectively, $ D_\alpha $ in the regime where $ \alpha\geq 0 $ may be simplified to
\begin{equation}\label{eq:div_bd}
D_\alpha (\rho_S\|\tau_{S_\beta}) = \frac{1}{\alpha-1} \ln \sum_i p_i^\alpha\tau_i^{1-\alpha}.
\end{equation}
The reader who is familiar with R{\'e}nyi divergences knows that smooth variants, denoted as $ D_\alpha^\varepsilon $ have long existed \cite{renner2004smooth,renner2008security,datta2009min}, and have been shown to also converge to the relative entropy \cite{2ndlaw}, which recovers the Helmholtz free energy when substituted into Eq.~\eqref{eq:alphaFE}. Therefore, why not simply replace $ D_\alpha $ with $ D_\alpha^\varepsilon $? The reason why such an approach is undesirable can be seen from the form of these quantities\footnote{From now on, we drop the subscript $ S $ from the states such as $ \rho_S,\tau_S $, since in the rest of the paper they refer to the system by default; subscripts are used only when other systems such as the bath, or the catalyst are discussed. }:
\begin{equation}\label{eq:original_smootheddiv}
D_\alpha^\epsilon(\rho||\tau_\beta) = \begin{cases}
\displaystyle\max_{\tilde{\rho}\in\mathcal{B}^\varepsilon (\rho)} D_\alpha(\tilde{\rho}||\tau_\beta) & \text{ if } 0\leq\alpha\leq1,\\
\displaystyle\min_{\tilde{\rho}\in\mathcal{B}^\varepsilon (\rho)} D_\alpha(\tilde{\rho}||\tau_\beta) & \text{ if } \alpha>1,
\end{cases}
\end{equation}
where the optimization in Eq.~\eqref{eq:original_smootheddiv} is over the set of all quantum states $ \varepsilon $-close in terms of trace distance to $ \rho $, denoted as $ \mathcal{B}^\varepsilon (\rho)$. 
Note that for different regimes within $ \alpha\geq 0 $, the optimization is different (min/max), and moreover, the solution $\tilde{\rho}_\alpha$ would be in general dependent on $ \alpha $. Therefore, when jointly comparing 
$ D_\alpha^\varepsilon(\rho\|\tau_\beta)$ and $D_\alpha^\varepsilon(\sigma\|\tau_\beta) $ for all $ \alpha $, the operational meaning of comparing these divergences remains unclear, since it does not directly imply the comparison between divergences of a specific initial and final state $ \rho_\varepsilon,\sigma_\varepsilon$, and thus the second laws~\cite{2ndlaw} cannot be applied, except solely in the limit where $ \varepsilon\rightarrow 0 $.
On the other hand, the construction of our \gfe~involve the replacement of $ D_\alpha $ with $ \hat D_\alpha^\varepsilon $, that depends on explicit constructions of two block-diagonal states $ \flattest,\steep $, which we call the flattest state and the steep state:
\begin{equation} \label{eq:newdiv}
\hat{D}_\alpha^\epsilon(\rho||\tau_\beta) = \begin{cases}
D_\alpha(\steep||\tau_\beta)	& \text{ if } 0\leq\alpha\leq1,\\
D_\alpha(\flattest\|\tau_\beta)	& \text{ if } \alpha>1.
\end{cases}
\end{equation}
The explicit construction of $ \flattest,\steep $ that we use here can be found in Section \ref{sec:steepflat}, and it is such an explicit construction that makes it possible to have an operational meaning in terms of state transitions. 
Here, we leave one remark about these states, in order to motivate such a definition. The state $ \flattest $ is special in the sense that any other state $\rho'\in\epsball(\rho)$ (including non-block diagonal states) can always be transformed to $\flattest$ by thermal operations (TO)~\cite{HO13}, which is simply a special case of catalytic thermal operations where the catalyst is not needed. This can be expressed in terms of exact R\'enyi divergences: for all $\alpha\geq 0$, and any $ \varepsilon \geq 0 $, if $ \rho'\in \epsball(\rho) $, then 
\begin{equation}\label{eq:div_flattest}
D_\alpha(\rho'||\tau_\beta) \geq D_\alpha(\flattest||\tau_\beta).
\end{equation}
In particular, since we constructed $ \steep $ such that $ \steep \in\epsball(\rho)$, this means that 
$ D_\alpha(\steep||\tau_\beta) \geq D_\alpha(\flattest||\tau_\beta), $
and therefore the steep state can always be transformed to the flattest state. However, the steep state $ \steep $ does not enjoy the same kind of uniqueness as $ \flattest $; we later prove that one cannot always find a unique candidate for $ \steep $ that can be transformed to any state $ \rho' \in \epsball (\rho)$.

We can make use of the properties of $ \flattest $ and $ \steep $ to prove the operational meaning of the smoothed quantities in Eq.~\eqref{eq:newdiv}. 
By defining new smooth \gfe~as
\begin{equation}\label{eq:defgenf}
\hat F_\alpha^\varepsilon(\rho,\tau_\beta):= \beta^{-1}\cdot \left[-\ln Z_\beta + \hat D_\alpha^\varepsilon (\rho\|\tau_\beta)\right],
\end{equation}
we may state our main result as Theorem \ref{thm:equiv}.
\begin{theorem} \label{thm:equiv}
	Consider two states $\rho$ and $\sigma$ block-diagonal with respect to the Hamiltonian $ H $. Let $ \tau_\beta $ be the thermal state at inverse temperature $ \beta $, where $ \beta > 0 $. If for all $\alpha\geq 0$, we have
	\begin{equation}\label{key}
	\hat F_\alpha^\varepsilon(\rho,\tau_\beta) \geq \hat F_\alpha^\varepsilon(\sigma,\tau_\beta),
	\end{equation}
	then the exact state transition $\steep\xrightarrow[\rm CTO]~\sigma_{\rm fl}^\varepsilon$ is possible by a catalytic thermal operation.
\end{theorem}

There are two remaining questions. Firstly, how do the smoothed quantities $ D_\alpha^\varepsilon $ and $ \hat D_\alpha^\varepsilon $ relate to each other? We find that for any $ \varepsilon \in [0,1] $, an explicit state $ \flattest $ always exists such that Eq.~\eqref{eq:div_flattest} holds. Therefore, we know that the minimizations in Eq.~\eqref{eq:original_smootheddiv} are obtained by $ \flattest $. 
Using this property, we may rewrite the conventional smoothed R\'enyi divergences as
\begin{equation} \label{eq:olddiv}
D_\alpha^\epsilon(\rho||\tau) = \begin{cases}
\displaystyle\max_{\tilde{\rho}\in\mathcal{B}^\varepsilon (\rho)} D_\alpha(\tilde{\rho}||\tau)	& \text{ if } 0\leq\alpha\leq1\\
D_\alpha(\flattest||\tau)	& \text{ if } \alpha>1.
\end{cases}
\end{equation}
This shows that these new smoothed divergences are quite similar to the original smoothed divergences: for $\alpha>1$, they are equivalent. 
However, the same is no longer true for $ 0\leq\alpha\leq 1 $, i.e. we show that it is not possible to replace the maximization in Eq.~\eqref{eq:olddiv} with a single explicit state. This is also why Theorem \ref{thm:equiv} is only a sufficient condition (but not necessary); there can be multiple candidates in $ \mathcal{B}^\varepsilon (\rho) $ which are steeper than $ \rho $, but maximize $ D_\alpha $ for different values of $ \alpha $. For a particular state transition, the best $ \steep $ candidate may depend on the final target state.

The second question is whether the \gfe~in Eq.~\eqref{eq:defgenf} recover the macroscopic second law when approaching the thermodynamic limit. We show that this is true, by proving that our new smoothed quantities satisfy the asymptotic equipartition property:
\begin{theorem}\label{thm:AEP}
	Consider any state $ \rho $ block-diagonal with respect to the Hamiltonian $ H $. Then for all $\alpha\geq 0$,
	\begin{equation}
	\lim_{\varepsilon\rightarrow0}\lim_{n\rightarrow\infty}\frac{1}{n}\hat{F}_\alpha^\epsilon(\rho^{\otimes n},\tau_\beta^{\otimes n}) = F(\rho,\tau_\beta).
	\end{equation}	
\end{theorem}
In proving Theorem \ref{thm:AEP}, we obtain explicit upper and lower bounds (see Appendix \ref{app:C_AEP}) of the form
\begin{equation}\label{key}
F(\rho,\tau_\beta) - f(n,\varepsilon) \leq \frac{1}{n}F_\alpha^\varepsilon (\rho^{\otimes n},\tau_\beta^{\otimes n}) \leq F(\rho,\tau_\beta) + g(n,\varepsilon),
\end{equation}
where one can show that $ f(n,\varepsilon) $ and $ g(n,\varepsilon) $ vanish in the limits $ n\rightarrow\infty $ and $ \varepsilon\rightarrow 0 $\footnote{The functions $ f $ and $ g $ as shown in Appendix \ref{subsec:proofT2}, have an implicit dependency on $ \rho $ and $ \tau $ as well. However, for any $ \rho $ and $ \tau $ (thermal state), we can show that these functions vanish in the desired limits $ n\rightarrow\infty $ and $ \varepsilon\rightarrow 0 $.}. Furthermore, these bounds are still useful should one be interested in finite values of $ n $ and $ \varepsilon $. This is in contrast to Ref.~\cite{2ndlaw}, where when using the previously known quantities $ D_\alpha^\epsilon $ in Eq.~\eqref{eq:original_smootheddiv}, one can only recover the macroscopic second law in the limit $ n\rightarrow\infty $ and $ \varepsilon\rightarrow 0$, while for finite $ n,\varepsilon $, there is no operational meaning in terms of state transitions. Our results also show that for finite values of $ n$ and $\varepsilon $, one can easily check whether there exists a particular approximate transition: if
\begin{equation}\label{eq:finiten}
F(\rho,\tau_\beta)	\geq F(\sigma,\tau_\beta) + \beta^{-1}\Delta(n,\varepsilon,\rho,\sigma,\tau_\beta),
\end{equation}
then $ (\rho^{\otimes n})_{\rm steep}^\varepsilon \rightarrow  (\sigma^{\otimes n})_{\rm fl}^\varepsilon$ is possible via thermal operations. The explicit form of $ \Delta(n,\varepsilon,\rho,\sigma,\tau_\beta) $ is derived in Corollary \ref{cor:onecondition} in Appendix \ref{subsec:proofT2}, and vanishes to zero in the limit $ \varepsilon\rightarrow 0 $ and $ n\rightarrow\infty $. Such a bound is useful for example in the following situation: consider $ \rho $ and $ \sigma $ such that we know $ F(\rho,\tau_\beta) > F(\sigma,\tau_\beta) $, and therefore in the thermodynamic limit, one can asymptotically transform $ n $ copies of $ \rho $ into $ \sigma $ via CTOs. However, it is possible that when one considers a single-copy transformation, Eq.~\eqref{eq:exactgenfree} is not satisfied for all $ \alpha\geq 0 $, and therefore the transition cannot take place. However, one can use Eq.~\eqref{eq:finiten} to find a lower bound such that whenever $n \geq n^* $, then $ (\rho^{\otimes n})_{\rm steep}^\varepsilon \rightarrow  (\sigma^{\otimes n})_{\rm fl}^\varepsilon$ is possible, by invoking $ \Delta (n^*,\varepsilon,\rho,\sigma,\tau_\beta) \leq \beta [F(\rho,\tau_\beta)-F(\sigma,\tau_\beta)]$.
\section{Steep and flat states}\label{sec:steepflat}
\subsection{Motivation and Definition}

\begin{figure*}
	\centering
	\begin{subfigure}{.49\textwidth}
		\centering
		\includegraphics[width=.92\linewidth]{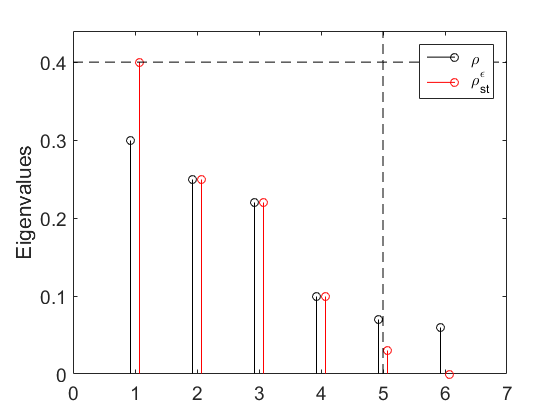}
		\vspace{-0.7cm}\subcaption{}
		\label{fig:steepest}
	\end{subfigure}%
	\hspace{0.2cm}
	\begin{subfigure}{.49\textwidth}
		\centering
		\includegraphics[width=.92\linewidth]{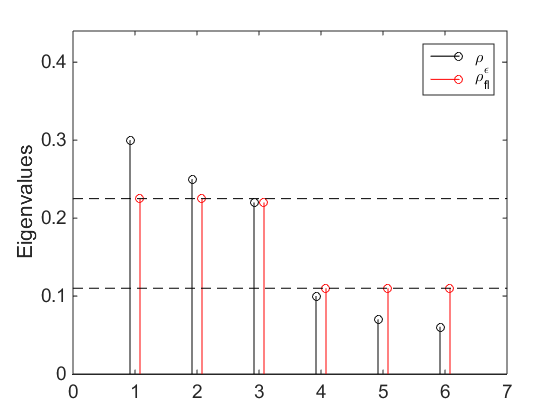}
		\vspace{-0.7cm}\subcaption{}
		\label{fig:flattest}
	\end{subfigure}
	\caption{\small Steepest and flattest states of $ \rho $ with ordered eigenvalues $ \eig(\rho) = \lbrace 0.3,
		0.25,   0.22,  0.1, 0.07, 0.06 \rbrace $, when the Hamiltonian is trivial and $ \varepsilon=0.1 $. In Fig.~\ref{fig:steepest}, the steepest state is obtained by cutting the distribution tail, and increasing the largest eigenvalue to normalize. Therefore, we have $  \eig(\steepest) = \lbrace 0.4,
		0.25,   0.22,  0.1, 0.04, 0 \rbrace $: all eigenvalues to the right of the vertical line are cut, and $ \varepsilon $ is added to the first eigenvalue. In Fig.~\ref{fig:flattest}, the flattest state is constructed by cutting the largest eigenvalues up to $\epsilon$. One visualizes this as having an upper dashed, horizontal line gradually lowered until the probability mass laying above equals $ \varepsilon $. This mass is cut and redistributed by adjusting the lower dashed, horizontal line to a height, such that if one increases all probabilities laying below this line (i.e. 4-6 in this figure), up to this line, a total of $ \varepsilon $ is added. This gives $ \eig(\flattest) = \lbrace 0.225, 0.225, 0.22,  0.11,
		0.11, 0.11 \rbrace$. }
\end{figure*}
\begin{figure}
	\includegraphics[width=0.97\linewidth]{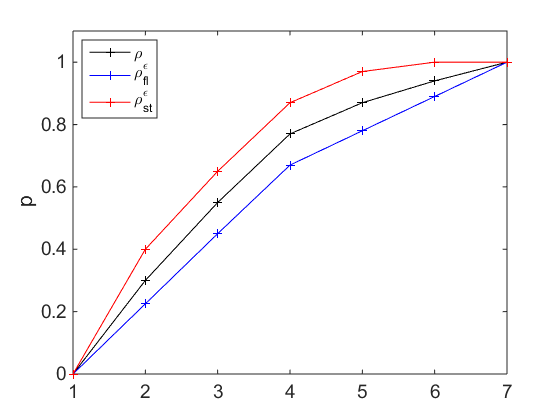}
	\caption{\small Majorization curves of a state $ \rho $ compared to its steepest and flattest states $ \steepest $ and $ \flattest $. Any other state $ \rho'\in\epsball (\rho) $ will have a majorization curve that lies between the red and blue curve.}\label{fig:TMC}
\end{figure}
Here, we present explicit smoothing procedures used in the definition of $ \hat D_\alpha^\varepsilon $ given in Eq.~\eqref{eq:newdiv}. Given a quantum state denoted by $\rho$, and a smoothing parameter $\epsilon>0$, we would like to find the most ``advantageous'' or ``disadvantageous'' states that are close to $\rho$ in terms of trace distance. By most advantageous, we mean that the state may reach as many other states that are also close to $\rho$ as possible. Similarly, by most disadvantageous, we mean that such a state may always be obtained from other states which are also close to $ \rho $. 

We find these states by considering transitions via thermal operations (TO)~\cite{BMORS13,HO13}, which are CTOs without a catalyst: in the description given in Eq.~\eqref{eq:CTO}, the system $ C $ is dropped completely. 
Our analysis is focused on the subset of states which commute with the Hamiltonian. Note that TOs form a subset of CTOs, so if a transition can be performed with a TO, then the transition can also be performed by a CTO. 
To find these states, we will mainly be analyzing \emph{thermo-majorization curves}, which is the necessary and sufficient condition that determines the possibility of a transition $ \rho\xrightarrow[\rm TO]~\rho' $~\cite{HO13}. 

Consider a block-diagonal quantum state $\rho$ associated with a Hamiltonian $ H $. Given the set $ \epsball(\rho) $, consider a special subset of block-diagonal states $ \epsball_D(\rho) \subseteq \epsball(\rho) $, 
with
\begin{equation}
\epsball_D(\rho) = \{\rho'| \rho'\in\epsball(\rho), [\rho',H] = 0\}.
\end{equation}
If a state in $ \epsball_D(\rho) $ is more advantageous than $ \rho $, we call this an $\epsilon$-steep state; similarly if it is less advantageous, we call this an $\epsilon$-flat state. In particular, we use the following terminology: a block diagonal state $\hat{\rho}$ is $\epsilon$-steeper than $\rho$ if $\hat{\rho}\in\epsball_D(\rho)$ and $\hat{\rho}\rightarrow\rho$ is possible via thermal operations. On the other hand, we say that a block diagonal state $\tilde{\rho}$ is $\epsilon$-flatter than $\rho$ if $\tilde{\rho}\in\epsball_D(\rho)$ and $\rho\rightarrow\tilde{\rho}$ is possible via thermal operations.
We leave two remarks about these definitions. First of all, it should be noted that not all states in $\epsball(\rho)$ satisfy either of these definitions; there exist incomparable states pairs $ \rho, \bar\rho $ where the transition cannot happen either way. 
Secondly, we can compare the R\'enyi divergence of these $\epsilon$-steep and $\epsilon$-flat states. For an $\epsilon$-steep state $\hat{\rho}$, because the transition $\hat{\rho}\xrightarrow[\rm TO]{}\rho$ is possible, the transition $\hat{\rho}\xrightarrow[\rm CTO]{}\rho$ is possible as well. Similarly, for any $\epsilon$-flat state $\tilde{\rho}$, the transition $\rho\xrightarrow[\rm CTO]{}\tilde{\rho}$ is possible. Thus, we know that their R\'enyi divergences satisfy for $\alpha\geq 0$,
\begin{equation}
D_\alpha(\hat{\rho}||\tau) \geq D_\alpha(\rho||\tau) \geq D_\alpha(\tilde{\rho}||\tau).
\end{equation}

Next, we look at extreme cases of $\epsilon$-steep and $\epsilon$-flat states, which we refer to as the $\epsilon$-steepest and $\epsilon$-flattest states.
\begin{definition}\label{def:steepest}
	The block-diagonal state $\steepest$ is the $\epsilon$-steepest state if $\steepest\xrightarrow[\rm TO]{}\rho'$ is possible for any $\rho'\in\epsball_D(\rho)$, or in other words, $ \steepest $ thermo-majorizes $ \rho' $.
\end{definition}
\begin{definition}\label{def:flattest}
	The block diagonal state $\flattest$ is the $\epsilon$-flattest state if the transition $\rho'\xrightarrow[\rm TO]{}\flattest$ is possible for any $\rho'\in\epsball(\rho)$, or in other words, $ \rho' $ thermo-majorizes $ \flattest $.
\end{definition}
As mentioned above, not all states are comparable when considering arbitrary Hamiltonians. This implies that $ \steepest $ and $ \flattest $ do not necessarily always exist for any $ \varepsilon$, introducing
additional challenges. To get some intuition, let us first mention however that they always exist for the simplest case of fully-degenerate (trivial) Hamiltonians (see \cite{steepflat} for proofs, and application in \cite{nilanjana} to study continuity bounds). A visual construction is shown in Figs.~\ref{fig:steepest} and \ref{fig:flattest}, and the reader may refer to Appendix \ref{app:B1_trivH} for the explicit mathematical construction. Fig.~\ref{fig:TMC} shows the majorization curve for $ \steepest $ and $ \flattest $, in comparison with $ \rho $. For general Hamiltonians, thermo-majorization curves have to be compared instead, and this complicates the task of finding steepest and flattest states, because the kinks do not align in their horizontal position (in contrast to Fig.~\ref{fig:TMC}).

\subsection{Constructing the flattest state and an $\epsilon$-steeper state for general Hamiltonians}\label{subsec:construction}
Let us turn to more general Hamiltonians with discrete energy levels. It is no longer straightforward to find the $ \varepsilon $-steepest or flattest states, 
because the optimal smoothing strategy depends on the Hamiltonian. Nevertheless, we can show that the $ \varepsilon $-flattest state always exists, by providing an explicit method to construct $ \flattest $. Consider a $d$-dimensional state $\rho$ block-diagonal in the energy eigenbasis, and write down its eigenvalues $\{p_i\}_i$ in a $\beta$-ordered form, such that
\begin{equation}\label{key}
p_1 e^{\beta E_1}\geq \cdots\geq p_d e^{\beta E_d}.
\end{equation}
For a smoothing parameter $ \epsilon$, the flattest state of $\rho$ can be constructed as follows: If $\epsilon$ is large enough, such that the trace distance $ \delta(\rho,\tau_\beta) \geq \varepsilon$, then we know that $ \tau_\beta\in\bdepsball(\rho) $. Since all states may go to $ \tau_\beta $ via thermal operations, by definition
the flattest state is equal to the thermal state. 
Otherwise, if $ \delta(\rho,\tau_\beta) < \varepsilon$, the construction involves determining certain indices $ M,N $ where $1\leq M\leq N \leq d$. These indices tell us which eigenvalues of $\rho$ we have to modify. In particular, let $M$ be the smallest integer such that 
\begin{equation}
\epsilon \leq \sum_{i=1}^{M}p_{i} - p_{M+1}e^{\beta E_{M+1}}\sum_{i=1}^{M}e^{-\beta E_{i}}.
\end{equation}
Similarly, let $N$ be the largest integer such that
\begin{equation}
\epsilon \leq  p_{N-1}e^{\beta E_{N-1}}\sum_{i=N}^{d}e^{-\beta E_{i}} - \sum_{i=N}^{d} p_{i}.
\end{equation}
We prove in Lemma \ref{lem6}, Appendix \ref{app:D_TL} that $ M \leq N$. The flattest state can then be constructed by cutting the first $ M $ eigenvalues $ \lbrace p_i\rbrace_{i=1}^M $ by a total amount of $ \varepsilon $, and increasing the eigenvalues $ \lbrace p_i\rbrace_{i=N}^d $ by another $ \varepsilon $ for renormalization. Moreover, the eigenvalues are cut/increased in such a way that $ \tilde p_1 e^{\beta E_1} = \cdots = \tilde p_M e^{\beta E_M} $, and similarly $ \tilde p_N e^{\beta E_N} = \cdots = \tilde p_d e^{\beta E_d} $. This construction means that $ \flattest $ not only is diagonal in the same basis as $ \rho $ itself, it also has the same $ \beta $-ordering.
Given these indices, the eigenvalues of $ \flattest $ are given by
\begin{equation}
\tilde{p}_i = \begin{cases}
e^{-\beta E_i}\frac{\left(\sum_{i=1}^{M} p_i\right)-\epsilon}{\sum_{i=1}^{M} e^{-\beta E_i}} & \text{if } i\leq M\\[8pt]
e^{-\beta E_i}\frac{\left(\sum_{i=N}^{d} p_i\right)+\epsilon}{\sum_{i=N}^{n} e^{-\beta E_i}} & \text{if } i\geq N\\
p_i & \text{otherwise}.
\end{cases}
\end{equation}

Unfortunately, a similar construction does not exist for the steepest state. In particular, we prove that at least for some states $ \rho $ and parameters $ \varepsilon >0 $, $ \steepest $ as defined in Def.~\ref{def:steepest} does not exist. Therefore, we give a way to construct a particular $ \varepsilon $-steep state $ \steep $ instead: 
if $\epsilon>1-p_1$, then the eigenvalues $\{\hat{p}_i\}_i$ of the steep state are given by
\vspace{-0.25cm}\begin{equation}
\hat{p}_i = \begin{cases}
1 & \text{if } i=1\\
0 & \text{otherwise}.
\end{cases}
\end{equation}

\vspace{-0.25cm}

For any $ 0<\epsilon\leq 1-p_1$, we cannot reach this pure state. Therefore, we need to find the eigenvalues that we can cut while remaining within the $\epsilon$-ball.
We do this by first choosing the index $R\in\mathbb{N}$ such that $\sum_{i=R}^{d}p_i\geq \epsilon > \sum_{i=R+1}^{d}p_i$. Then, we define $ \steep $ to be the state diagonal in the same basis as $ \rho $, with the eigenvalues
\vspace{-0.25cm}\begin{equation}
\hat{p}_i = \begin{cases}
p_1+\epsilon & \text{if } i=1\\
p_i & \text{if } 1<i<R\\
p_i+\sum_{i=R+1}^{d}p_i-\epsilon & \text{if } i=R\\
0 & \text{otherwise}.
\end{cases}
\end{equation}

\vspace{-0.35cm}

\subsection{Proof of Theorem \ref{thm:equiv}}
Once the flattest and steep states are established in Section \ref{subsec:construction}, we can spell out the proof of our main result.
 
\noindent \begin{myproof}{Theorem \ref{thm:equiv}}
For states $ \rho,\tau $, and a particular $ \varepsilon >0 $ assume that $\hat{D}_\alpha^\epsilon(\rho||\tau) \geq \hat{D}_\alpha^\epsilon(\sigma||\tau)$ for all $ \alpha\geq 0 $. Then, for $\alpha>1$ we have that
\begin{align}
	D_\alpha(\steep||\tau) \geq
	D_\alpha(\flattest||\tau)
	& = \hat{D}_\alpha^\epsilon(\rho||\tau)\\
	 &\geq \hat{D}_\alpha^\epsilon(\sigma||\tau)
	 = D_\alpha(\sigma^\epsilon_\textup{fl}||\tau).\nonumber
\end{align}	
For $0\leq\alpha\leq1$ we have that		
\begin{align}
	D_\alpha(\steep||\tau) = \hat{D}_\alpha^\epsilon(\rho||\tau)
	 &\geq \hat{D}_\alpha^\epsilon(\sigma||\tau)\\
	 &= D_\alpha(\sigma^\epsilon_\textup{steep}||\tau)
	 \geq D_\alpha(\sigma^\epsilon_\textup{fl}||\tau).\nonumber
\end{align}	
Thus, for all $\alpha\geq 0$ we have that the exact divergences $D_\alpha(\steep||\tau)\geq D_\alpha(\sigma^\epsilon_\textup{fl}||\tau)$. Therefore, the transition $\steep\rightarrow\sigma^\epsilon_\textup{fl}$ is possible via catalytic thermal operations by the second laws put forward in~\cite{2ndlaw}.	
\end{myproof}

\section{Discussion and conclusion}
The significance of thermo-majorization curves (TMC) go beyond the framework of thermal operations: these curves also constitute state transition conditions for a set of more experimental friendly processes called Crude Operations~\cite{perry2015sufficient}.~Moreover, it has also been shown that thermal operations are more powerful in enabling state transitions, when compared to protocols achieved mainly by weak thermal contact \cite{wilming2016second}; for example they allow anomalous heat flow, which is a larger change in temperature than allowed if one only considers weak thermal contact with a heat bath. Because of its power, TMCs have been applied to study various problems in thermodynamics, such as work extraction~\cite{renes2014work,HO13}, heat engines efficiencies~\cite{woods2015maximum,surpassCarnot,chubb2017beyond} cooling rates \cite{masanes2014derivation,wilming2017third} and thermodynamic reversibility \cite{chubb2017beyond} in the quantum regime. 
In our work, we proposed newly defined smoothed generalized free energies; this has been achieved by understanding how to construct smoothed states that have optimal advantage/disadvantage under thermal operations. In the process, we developed technical bounds on the difference between two TMCs (Appendix \ref{app:A}, Theorem \ref{theorem:epsborder}), as a function of the trace distance between two states (Fig.~\ref{fig:boundsTMC}). Previously, thermo-majorization was hard to analyze because even when comparing two states close in trace distance, they might have completely different $ \beta $-orderings, arising to different shapes in their TMC. However, our bounds hold solely as a function of trace distance, irrespective of the $ \beta $-ordering. Therefore, these bounds might be of general use when analyzing TMCs.

The scope of our work has been restricted to block-diagonal states. For arbitrary state transitions, even the necessary and sufficient conditions for exact transitions are unknown \cite{2ndlaw,lostaglio2015description,lostaglio2015quantum}, and remain a large open problem in quantum thermodynamics (thermo-majorization, however, remains a necessary condition \cite{lostaglio2015quantum}). The case for a single qubit has been solved in \cite{cwiklinski2015limitations}, which may be a starting point to consider optimal smoothing that takes coherence into account. Alternatively, one may also choose to investigate a larger set of thermal processes compared to thermal operations, such as Gibbs preserving maps \cite{faist2015gibbs,faist2017fundamental} or generalized thermal processes \cite{gour2017quantum}. Such processes recover thermo-majorization as the state transition condition when dealing with block-diagonal states, but for arbitrary quantum states, they achieve a strictly larger set of state transitions when compared to thermal operations. Very recently, necessary and sufficient conditions for state transitions have been identified for both types of processes \cite{faist2017fundamental,gour2017quantum}. Comparison between optimal smoothing procedures for these various different processes could potentially help us to understand their fundamental differences.

\acknowledgements
We thank Renato Renner and Mischa Woods for interesting discussions, and Kamil Korzekwa for discussions and remarks on the manuscript.
RM, NN and SW were supported by STW Netherlands, and NWO VIDI and an ERC Starting Grant. 


\clearpage

\appendix
This appendix provides the full derivation of technical details used to obtain our main results. In Appendix \ref{app:A}, we recall the definition of thermal operations and thermo-majorization in full. We develop a useful tool in this section concerning generalized curves that resemble the form of thermo-majorization curves. Using this tool, we show that the distance between thermo-majorization curves of two block-diagonal states may be bounded by their trace distance.

Appendix \ref{app:B} presents the constructions of flattest and steepest states. In Appendix \ref{app:B1_trivH}, we start by proving that such states always exist for the trivial Hamiltonian. For general Hamiltonians, the flattest and steepest states are investigated accordingly in Appendices \ref{app:B2flattest} and \ref{app:B3steepest}. Certain technical Lemmas used in Appendix \ref{app:B2flattest} were proven later on in Appendix \ref{app:D_TL}.

Lastly, in Appendix \ref{app:C_AEP} we prove the asymptotic equipartition property for our new divergences.	

\section{Thermo-majorization and some technical tools}\label{app:A}
In this section, we introduce the tools necessary to derive the results stated in the main text of this manuscript. We start by defining the notion of thermo-majorization curves for states which are block-diagonal in the energy eigenbasis, and present a few lemmas that will be useful in deriving the main results on steepest and flattest states.

To model these thermodynamic operations, we adapt the paradigm of thermodynamic resource theories, where state transitions are achieved via \emph{thermal operations}~\cite{BMORS13,HO13}. A thermal operation on some quantum system $ S $ is defined by two elements:
\begin{enumerate}
	\item a bath of some fixed inverse temperature $ \beta $, which is a quantum state of the form \begin{equation}\label{key}
	\tau_{B_\beta} = \frac{1}{\tr\left( e^{-\beta H_B}\right)}e^{-\beta H_B}.
	\end{equation}
	\item a unitary $ U_{SB} $ that preserves the total energy of the global system $ SB $, i.e. the commutator $ [ U_{SB} ,  H_{SB} =0] $, where $  H_{SB} =  H_S \otimes \id_B + \id_S\otimes H_B$.
\end{enumerate}

When one considers only initial states $ \rho_S $ that are block-diagonal in the energy eigenbasis, then necessary and sufficient conditions for state transition to occur via thermal operations are given by thermo-majorization, which we will soon explain. However, as mentioned in the main text, for catalytic thermal operations, the conditions on the free energies $F_\alpha( \rho_S,\tau_S)$ fully determine whether or not a state transition is achievable or not. 
Since thermal operations form a special subset of catalytic thermal operations, we therefore know that thermo-majorization is a more stringent condition compared to the free energies.

The thermo-majorization curve of a state $ \rho $ which is block-diagonal with respect to its corresponding Hamiltonian $  H $ determines the set of final states achievable via thermal operations: Any block diagonal state which has a thermo-majorization curve that lies below the curve of $\rho$ can be reached.
For a $d$-dimensional state $\rho = \sum_i p_i \proj{E_i}$ that is diagonal in the energy eigenbasis, we first denote $ p=\lbrace p_1,\cdots,p_d\rbrace $ to be a vector containing the eigenvalues of $ \rho $, which are the occupational probabilities corresponding to energy levels given in the vector $ E=\lbrace E_1,\cdots,E_d\rbrace $. Subsequently, let $\hat p = \{\hat p_1,\dotsc,\hat p_d\}$ be a particular permutation of $ p $, with $\hat E = \{\hat E_1,\dotsc,\hat E_d\}$ being the same permutation upon $ E $. In particular, $ \hat p,\hat E $ is permuted in the ordering that 
\begin{equation}\label{eq:betaorder}
\hat p_1 e^{\beta \hat E_1} \geq \dotsc \geq \hat p_de^{\beta \hat E_d}.
\end{equation} 
It is helpful to note that although there might be several permutations that satisfy Eq.~\eqref{eq:betaorder} (for example, some inequalities might be satisfied with equality), these different permutations would give rise to the same thermo-majorization curve, so picking any permutation that satisfies Eq.~\eqref{eq:betaorder} suffices.
The energy spectrum $ \hat E $ also allows us to define the partition function for the system (of a certain temperature), which is given by $ Z = \sum_{i=1}^d e^{-\beta \hat E_i} $.
Given $ \hat p $ and $ \hat E $, the thermo-majorization curve is defined as the piecewise linear curve $c(\hat p,\hat E)$ which connects the points given by $\left\{\left(\sum_{i=1}^{k}e^{-\beta \hat E_i}/Z,\sum_{i=1}^{k} \hat p_i\right)\right\}_{k=0}^d$ with straight line segments. Due to the particular $ \beta $-ordering of $ \hat p $ and $ \hat E $, such a thermo-majorization curve is concave. 

In general, such a piecewiese-linear curve $ c(p,E) $ does not need to be defined only for the $ \beta $-ordered vectors $ \hat p, \hat E $, but for any permutation of the eigenvalues $ p,E $. In order to compare such curves, we use the notation $ c(p,E) \leq c(\hat{p},\hat{E}) $ to denote that $ c(p,E) $ lies completely below $ c(\hat{p},\hat{E}) $. We will also use the notation $ c(p,E) + \varepsilon $ to denote the piecewise linear curve that connects the points given by $\left\{\left(\sum_{i=1}^{k}e^{-\beta \hat E_i}/Z,\varepsilon + \sum_{i=1}^{k} \hat p_i\right)\right\}_{k=0}^d$. A special relation exists between any $ c(p,E) $ and the thermo-majorization curve $ c(\hat p,\hat E) $, which we detail in Lemma \ref{lemma:curve}. 

\begin{Lemma} \label{lemma:curve}
	Let $\rho$ be a $d$-dimensional system, with $d\in\mathbb{Z}^+$. Let $\hat p = \{\hat{p}_i\}_i$ be a vector containing the $\beta$-ordered eigenvalues of $\rho$, with $\hat E = \{\hat{E}_i\}_i$ containing the corresponding energy levels. Let $p$ be any other vector which is an arbitrary permutation of the entries in $\hat p$, with $ E $ being the same permutation of $ \hat E $. Then, $c(p,E)\leq c(\hat{p},\hat{E})$.
\end{Lemma}
\begin{proof}
	Since we want to prove the above lemma for an arbitrary permutation of $p$ and $E$, let us consider two possible scenarios. In the first case, $ p,E $ is also $ \beta $-ordered, i.e. they satisfy
	\begin{equation}\label{eq:border}
	p_1e^{\beta E_1} \geq \dotsc \geq p_de^{\beta E_d}.
	\end{equation}
	Note that this happens either when the permutation is trivial, i.e. $ \hat p = p $ (and $ \hat E = E $); or it is also possible that certain inequalities in Eq.~\eqref{eq:betaorder} are achieved with equality, so that the $ \beta $-ordering is not unique.
	The curves $c(p,E)$ and $c(\hat{p},\hat{E})$ will be the same in these cases, such that $c(p,E)\leq c(\hat{p},\hat{E})$ holds trivially.
	
	The second case is that $ p,E $ now do not satisfy Eq.~\eqref{eq:border}, i.e., they are not yet $\beta$-ordered. This implies, that we can find at least one index $n$ such that $p_ne^{\beta E_n}<p_{n+1}e^{\beta E_{n+1}}$. Intuitively, such a relation means that when the curve $ c(p,E) $ is drawn, then 
	$ c(p,E) $ will be convex (instead of being concave) in the interval $ (\sum_{i=1}^{n-1} e^{-\beta \hat E_i}/Z, \sum_{i=1}^{n+1} e^{-\beta \hat E_i}/Z) $.
	We will now analyze the curve $ c(p,E) $ more closely around such a point. 
	
	To do so, we define the vectors $ \tilde p,\tilde E$ such that
	\begin{equation}\label{eq:singleswap}
	\tilde{p}_i = \begin{cases}
	p_{n+1} & \text{ if } i=n,\\
	p_n & \text{ if } i=n+1,\\
	p_i & \text{ otherwise,}
	\end{cases}
    \end{equation}
	and
    \begin{equation}
	\tilde{E}_i = \begin{cases}
	E_{n+1} & \text{ if } i=n,\\
	E_n & \text{ if } i=n+1,\\
	E_i & \text{ otherwise.}
	\end{cases}
	\end{equation}
%
%
	
	If we then compare $c(p,E)$ with $c(\tilde{p},\tilde{E})$, we see that for the points
	\begin{eqnarray}\label{key}
	A = (x_A,y_A) = \left(\sum_{i=1}^{n-1}\frac{e^{-\beta E_i}}{Z},\sum_{i=1}^{n-1} p_i\right), \\ B = (x_B,y_B) = \left(\sum_{i=1}^{n+1}\frac{e^{-\beta E_i}}{Z},\sum_{i=1}^{n+1} p_i\right),
	\end{eqnarray}
	the curves completely overlap before the point $A$ and after the point $B$. However, the curves will differ within the $ x $-axis interval $ (x_A,x_B) $. We show that in this interval, the curve of $ c(\tilde p,\tilde E) $ will lay above that of $ c (p,E) $.
	To show this, note that both curves have exactly one kink in this region. 
	We will compare these kinks with the straight line through the points $A$ and $B$.
	To simplify the analysis, let us redefine the origin to be located at point $A$. The straight line through these two points is then given by
	\begin{align}
	y 
	 = \frac{(p_n+p_{n+1})Z}{e^{-\beta E_n} + e^{-\beta E_{n+1}}}x.
	\end{align}
	The kink of $c(p,E)$ is located at $(e^{-\beta E_n}/Z,p_n)$. The vertical height difference between the straight line and the kink, at $ x = e^{-\beta E_n}/Z $ is given by
	\begin{align}
	y-p_n & = \frac{(p_n+p_{n+1})Z}{e^{-\beta E_n} + e^{-\beta E_{n+1}}} \cdot \frac{e^{-\beta E_n}}{Z} - p_n\nonumber\\
	& = \frac{e^{-\beta E_n}(p_n+p_{n+1})-p_n(e^{-\beta E_n} + e^{-\beta E_{n+1}})}{e^{-\beta E_n} + e^{-\beta E_{n+1}}}\nonumber\\
	& = \frac{e^{-\beta E_n}p_{n+1}-e^{-\beta E_{n+1}}p_n}{e^{-\beta E_n} + e^{-\beta E_{n+1}}}\nonumber\\
	& = \frac{e^{\beta(E_n+E_{n+1})}}{e^{\beta(E_n+E_{n+1})}} \cdot \frac{e^{-\beta E_n}p_{n+1}-e^{-\beta E_{n+1}}p_n}{e^{-\beta E_n} + e^{-\beta E_{n+1}}}\nonumber\\
	& = \frac{e^{\beta E_{n+1}}p_{n+1}-e^{\beta E_{n}}p_n}{e^{\beta(E_n+E_{n+1})}(e^{-\beta E_n} + e^{-\beta E_{n+1}})}
	 > 0.
	\end{align}
	To summarize, we know that between the $ x $-axis interval $ (x_A,x_B) $, the following holds: 
	\begin{enumerate}
		\item The line $ y $ and the curve $ c(p,E) $ coincide at the points $ A $ and $ B $.
		\item The curve $ c(p,E) $ is piecewise-linear, and has a single kink in this interval which lies below the line $ y $.
	\end{enumerate}
	These two points imply that within the whole interval, $c(p,E)$ will lie below the straight line $ y $. 
	
	It is easy to see that the curve $ c(\tilde{p},\tilde{E}) $ will lie above the straight line, since it differs from $ c(p,E) $ only by a reordering of the two line segments, meaning that the two curves form a parallelogram. To prove this explicitly, note that the curve $c(\tilde{p},\tilde{E})$ has its kink located at $(e^{-\beta E_{n+1}}/Z,p_{n+1})$, and when we compare it with $ y $ at the position $ x = e^{-\beta E_{n+1}}/Z $, we find the opposite, i.e.
	\begin{align}
	y-p_{n+1} & = \frac{(p_n+p_{n+1})Z}{e^{-\beta E_n} + e^{-\beta E_{n+1}}} \cdot \frac{e^{-\beta E_{n+1}}}{Z} - p_{n+1}\nonumber\\
	& = \frac{e^{-\beta E_{n+1}}(p_n+p_{n+1})-p_{n+1}(e^{-\beta E_n} + e^{-\beta E_{n+1}})}{e^{-\beta E_n} + e^{-\beta E_{n+1}}}\nonumber\\
	& = \frac{e^{-\beta E_{n+1}}p_n-e^{-\beta E_n}p_{n+1}}{e^{-\beta E_n} + e^{-\beta E_{n+1}}}\nonumber\\
	& = \frac{e^{\beta(E_n+E_{n+1})}}{e^{\beta(E_n+E_{n+1})}}\cdot \frac{e^{-\beta E_{n+1}}p_n-e^{-\beta E_n}p_{n+1}}{e^{-\beta E_n} + e^{-\beta E_{n+1}}}\nonumber\\
	& = \frac{e^{\beta E_{n}}p_n-e^{\beta E_{n+1}}p_{n+1}}{e^{\beta(E_n+E_{n+1})}(e^{-\beta E_n} + e^{-\beta E_{n+1}})}\nonumber\\
	& < 0,
	\end{align}
	which means that by similar reasoning as before, in the region of interest,
	\begin{equation}
	c(\tilde{p},\tilde{E}) \geq y \geq c(p,E).
	\end{equation} 
	Thus, if we perform a swap between neighbouring elements of $p$, such that after swapping the elements $n$ and $n+1$ we have that $p_ne^{\beta E_n}\geq p_{n+1}e^{\beta E_{n+1}}$, then the new curve always lays above that of the old one.\\
	
	Using this, we can define a sequence of distributions $q^1,q^2,\cdots,q^d$ with corresponding energy levels $E^1,E^2,\cdots,E^d$, for any $ m\in\mathbb{Z}^+ $. 
	We define the sequence to start from $q^1=p$ and $E^1 = E$. Furthermore, for any $ n\geq 1 $, we obtain $q^{n+1}$ from $ q^n $ by a single swap. This swap is performed by the following procedure:
	\begin{enumerate}
		\item Identify the smallest index $ k $ such that $ q^n_k e^{\beta E_k^n} < q^n_{k+1} e^{\beta E_{k+1}^{n}}$.
		\item Obtain $ q^{n+1}, E^{n+1} $ from $ q^n,E^n $ by swapping the $ k $-th element with the $ k+1 $-th element. Such a swap is identical to the one we have seen in Eq.~\eqref{eq:singleswap}.
	\end{enumerate}
One can see that such a process is analogous to a bubble sort algorithm, where for finite dimension $ d $, there always exists an $ m\in\mathbb{Z}^+ $ large enough such that $ q^d = \hat{p} $ and $ E^d = \hat{E} $, i.e. the end result satisfies $ \beta $-ordering.	
	Therefore, for this sequence, we have that 
	\begin{eqnarray}
	c(p,E) &=& c(q^1,E^1)\nonumber\\
	&\leq& c(q^2,E^2)
	\leq\dotsc\leq c(q^d,E^d)\nonumber\\ 
	&=& c(\hat{p},\hat{E}).\nonumber
	\end{eqnarray} 
	This concludes the proof. 
\end{proof}

For any two states $ \rho,\sigma$, the trace distance $ \delta(\rho,\sigma) $ tells us how far apart the states are. For states which are diagonal in the same basis, if we denote $ p={\rm eig}(\rho) , q={\rm eig}(\sigma)$ as the corresponding eigenvalues, then
\begin{equation}\label{key}
\delta(\rho,\sigma) = \frac{1}{2} \sum_{i} |p_i-q_i|.
\end{equation}
The next theorem tells us how the thermo-majorization diagrams of block-diagonal states may behave, given an upper bound on their trace distance $ \varepsilon $. 
These bounds will be useful when we prove the optimality of steepest and flattest states in terms of thermo-majorization within the $ \epsilon $-ball of a state.
\begin{figure}[h!]
	\hspace{-0.9cm}\includegraphics[width=1.1\linewidth]{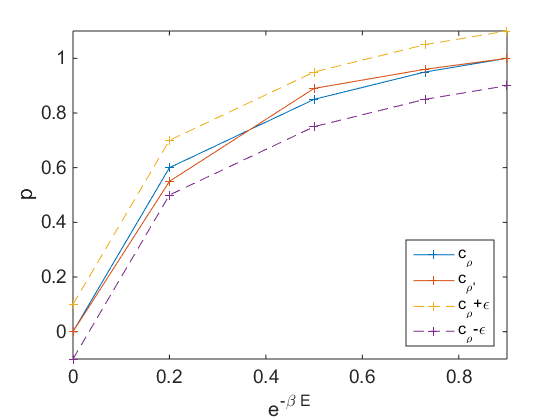}
	\caption{The thermo-majorization diagram of $\rho$ (blue) and the two bounds (yellow and purple). For any $ \rho'\in\epsball_D (\rho) $, its thermo-majorization curve must lie between the two bounds (as demonstrated with the red curve). These bounds are later used in Eq.~\eqref{eq:t3}. }\label{fig:boundsTMC}
	\label{epsborder}
\end{figure}
\begin{theorem} \label{theorem:epsborder}
	Consider any state $\rho$ block-diagonal with respect to some Hamiltonian $ H $, and any other $\rho' \in\bdepsball(\rho)$. Denote the thermo-majorization curves of $ \rho $ and $ \rho' $ as $ c_\rho $ and $ c_{\rho'} $ respectively.
	Then, as depicted in Fig.~\ref{epsborder},
	\begin{equation}\label{eq:t3}
	 c_{\rho}-\varepsilon \leq c_{\rho'} \leq c_{\rho} +\varepsilon.
	\end{equation}
\end{theorem}
\begin{proof}
	Let $p=\{p_i\}_i$ be the $\beta$-ordered eigenvalues of $\rho$ with corresponding energy levels $E = \{E_i\}_i$, such that $p_1e^{\beta E_1} \geq \dotsc \geq p_de^{\beta E_d}$. Therefore, the thermo-majorization curve of $ \rho $ is given by $ c_\rho = c(p,E) $. 
	On the other hand, let $p'=\{p'_i\}_i$ be the eigenvalues of $\rho'$; however, we do not write $ p' $ such that it is $ \beta $-ordered, instead we write it according to the same order as $p$. 
	Notice, therefore, that since $ p' $ is not necessarily $ \beta $-ordered, the thermo-majorization curve $ c_{\rho'} \neq c(p',E) $ in general.
	
	Because $\rho'\in\bdepsball(\rho)$, we have that the trace distance
	\begin{equation}
	\frac{1}{2}\sum_{i=1}^{d} \left|p_i-p'_i\right|\leq\epsilon. \label{eq:border1}
	\end{equation}
	Furthermore, because both states are normalized, we have that
	\begin{equation}
	\sum_{i=1}^{d} \left(p_i-p'_i\right)=0. 
	\end{equation}
	This means that 
	\begin{equation}
	\sum_{i=1}^{d} \left(p_i-p'_i\right) = 
	\sum_{i: p_i>p'_i} \left(p_i-p'_i\right) + 
	\sum_{i: p_i<p'_i} \left(p_i-p'_i\right) = 0,
	\end{equation}
	and thus 
	\begin{equation}\label{eq:equaltrd}
	\sum_{i: p_i>p'_i} \left(p_i-p'_i\right) = 
	-\sum_{i: p_i<p'_i} \left(p_i-p'_i\right),
	\end{equation}
	Applying Eq.~\eqref{eq:equaltrd} to Eq.~\eqref{eq:border1} yields
	\begin{align}
	\frac{1}{2}\sum_{i=1}^{d} \left|p_i-p'_i\right| & = 
	\frac{1}{2}\sum_{i: p_i>p'_i} \left(p_i-p'_i\right) - 
	\frac{1}{2}\sum_{i: p_i<p'_i} \left(p_i-p'_i\right)\nonumber\\
	 &= \sum_{i: p_i>p'_i} \left(p_i-p'_i\right)\nonumber\\
	 &= -\sum_{i: p_i<p'_i} \left(p_i-p'_i\right) \leq\epsilon.
	\end{align}
	We will consider two separate cases:\\[5pt]
	\textbf{(1) Both $ p$ and $p' $ have the same $\beta$-ordering.} In this case, we know that $ c_{\rho'} = c(p',E) $ holds, and the kinks of the two thermo-majorization curves $ c_\rho,c_{\rho'} $ line up. In this simple case, the maximum height difference between $ c_\rho $ and $ c_{\rho'} $ occurs at a kink, and therefore it is sufficient to compare the height of the curves at these discrete points. For any $  k\in\lbrace 1,d\rbrace $, at the $ k $-th kink which happens at the $ x $-coordinate $ x_k = Z^{-1}\sum_{i=1}^k e^{-\beta E_i}$, the height difference between the two curves is given by
	\begin{align}
	&\left|\sum_{i=1}^{k} p_i - \sum_{i=1}^{k} p'_i\right| \nonumber\\
	& = \left|\sum_{i\leq k: p_i>p'_i} \left(p_i - p'_i\right) + \sum_{i\leq k: p_i<p'_i} \left(p_i - p'_i\right)\right|\nonumber\\
	& = \left|\left|\sum_{i\leq k: p_i>p'_i} \left(p_i - p'_i\right)\right| - \left|\sum_{i\leq k: p_i<p'_i} \left(p_i - p'_i\right)\right|\right|\nonumber\\
	& \leq \max\left(\left|\sum_{i\leq k: p_i>p'_i} \left(p_i - p'_i\right)\right|,\left|\sum_{i\leq k: p_i<p'_i} \left(p_i - p'_i\right)\right|\right)\nonumber\\
	& \leq \epsilon. \label{eq:border2}
	\end{align}
	
	Thus, if $ \rho,\rho' $ have the same $\beta$-ordering of eigenvalues, then the height difference between $ c_\rho $ and $ c_{\rho'} $ cannot be larger than $\epsilon$.\\[5pt]	
	\textbf{(2) The states $ \rho $ and $ \rho' $ do not have the same $ \beta $-ordering. } 
	We can use the curve $ c(p',E) $ to show that the height difference between $ c_\rho $ and $ c_{\rho'} $ still cannot exceed $\epsilon$. By Lemma \ref{lemma:curve}, we know that $ c(p',E) \leq c_{\rho'} $. 
%
%
%
	Note that if we consider the curves $ c_\rho = c(p,E) $ and $ c(p',E) $, since $ p $ and $ p' $ have the same ordering, we know that the kinks of both curves always coincide. From case (1), we know that 
	\begin{equation}
	|c(p,E)-c(p',E)| \leq \epsilon,
	\end{equation}
and therefore	$ c(p',E) \geq c(p,E) - \epsilon $. Therefore, the thermo-majorization curve $ c_{\rho'} $ can also be lower-bounded by
	\begin{equation}\label{eq:crhoplowerbound}
	c_{\rho'} \geq c(p',E) \geq c(p,E) - \epsilon.
	\end{equation}
	
	
	Next, we need to prove that $c_{\rho'}\leq c(p,E) + \epsilon$ as well.
	This can be done with a similar strategy as before; except that we need to interchange the roles of $ p $ and $ p'$. In particular, let us first take the vectors $ p',E $ which were not $ \beta $-ordered, and denote $ q',E' $ to be the permuted versions of $ p',E $ such that $ q',E' $ now satisfies $ \beta $-ordering. More precisely, we use the permutation $ \Pi$ such that for $ q',E' $ defined by
	\begin{align}
	q_i' & = p_{\Pi(i)}'\\
	E_i' & = E_{\Pi(i)},
	\end{align}
	$ q',E' $ will now satisfy
	  \begin{equation}
	  q_1' e^{-\beta E_1'} \geq q_2' e^{-\beta  E_2'} \geq \cdots \geq q_d' e^{-\beta E_d'}.
	  \end{equation}
	 This implies that 
	 \begin{equation}\label{eq:thermomajrhop}
	  c_{\rho'} = c(q', E').
	 \end{equation}	 
	
	Now, similarly we may consider the permuted vector $ q = \Pi (p) $. Note that $ q,E' $ is a particular permutation of $ p,E $, so according to Lemma \ref{lemma:curve}, 
	\begin{equation}\label{eq:qEpeqcrho}
	c(q,E') \leq c(p,E)=c_{\rho}.
	\end{equation}
%
%
%
%
	Next, we will compare $c(q,E')$ with $c(q',E')$. First of all, note that since $ q= \Pi(p) $ and $ q' = \Pi(p') $, and since the trace distance is invariant under such permutations, we know that 
	\begin{equation}\label{key}
	\frac{1}{2} \sum_{i=1}^d |q_i - q_i'| \leq \epsilon
	\end{equation}
	holds as well. Also, since $ q $ and $ q' $ are both normalized vectors as well, the Eqns.~\eqref{eq:border1}-\eqref{eq:border2} hold for $q'$ and $q$.
	Since they are both ordered in the same way, the kinks of the two curves line up again at the same $ x $-coordinates, and therefore comparing the height of the curves at these coordinates will be sufficient. 
%
%
	Therefore, according to the analysis of case (1), the height difference $|c(q,E') - c(q',E')| \leq \epsilon$. Finally, combining this with Eq.~\eqref{eq:thermomajrhop} and Eq.~\eqref{eq:qEpeqcrho} allows us to conclude that
	\begin{align}\label{eq:crhopupperbound}
	c_{\rho'} = c(q',E') & \leq c(q,E') + \epsilon \leq c_{\rho} + \epsilon.
	\end{align}
Eq.~\eqref{eq:crhoplowerbound} and \eqref{eq:crhopupperbound} jointly prove the theorem for case (2).	
		
\end{proof}
Theorem \ref{theorem:epsborder} allows us to conclude the following: for any two block-diagonal states $ \rho,\rho' $ which are $ \epsilon $-close, regardless of whether $\beta$-ordering of the eigenvalues are same or different, the height difference between the thermo-majorization curves of $\rho$ and $\rho'$ cannot exceed $\epsilon$. Interestingly, the authors were made aware later on that a simpler proof can also be obtained by applying more general results in statistical literature, such as in Ref.~\cite{renes2016relative}.
This theorem gives us some bounds for the thermo-majorization curves of the states within the $\epsilon$-ball. Notice however, that the bounds cannot always be reached: in some regions, the lower bound can be negative, while in other regions the upper bound can also exceed $1$, as shown in Fig.~\ref{epsborder}. However, since eigenvalues form a normalized probability distribution, such bounds clearly cannot be reached.

\section{Flattest and steepest states}\label{app:B}
\subsection{Trivial Hamiltonians}\label{app:B1_trivH}

In this section, we will explain that for any smoothing parameter $\varepsilon >0$, for systems with trivial Hamiltonians, the steepest and flattest states always exist. We do so by providing the explicit construction of steepest and flattest states. A detailed proof of these constructions being steepest/flattest can also be found in \cite{steepflat}.

Consider an $m$-dimensional system $\rho$ with trivial Hamiltonian, and denote the ordered eigenvalues of $\rho$ as $\{p_i\}_i$. The eigenvalues $\{\hat p_i\}_i$ of the steepest state of $\rho$ are then given by 
\begin{equation}
\hat p_i = \begin{cases}
p_i+\epsilon & \text{if } i=1\\
p_i & \text{if } 1<i<M\\
p_i-\epsilon+\sum_{j=M+1}^{m}p_i & \text{if } i=M\\
0 & \text{otherwise,}\\
\end{cases}
\end{equation}
with $M\in\mathbb{N}$ such that
\begin{equation}
\sum_{i=M+1}^{m}p_i < \epsilon \leq \sum_{i=M}^{m}p_i.
\end{equation}
Here, we simply cut the tail of $\rho$, and added the cut probability mass to the first eigenvalue. This state majorizes all other states within the $\epsilon$-ball.\\

Consider the same state $\rho$, when $ \varepsilon < \delta(\rho, \id/m) $, where $ \id/m $ is the maximally mixed state. The eigenvalues $\{\tilde p_i\}_i$ of the flattest state of $\rho$ are then given by 
\begin{equation}
\tilde p_i = \begin{cases}
\frac{1}{N_1}\left(-\epsilon+\sum_{i=1}^{N_1}p_i\right) & \text{if } i\leq N_1\\
\frac{1}{m+1-N_2}\left(\epsilon+\sum_{i=N_2}^{m}p_i\right) & \text{if } i\geq N_2\\
p_i & \text{otherwise,}\\
\end{cases}
\end{equation} 
with $N_1\in\mathbb{N}$ such that
\begin{equation}
\sum_{i=1}^{N_1-1}\left(p_i-p_{N_1}\right) < \epsilon \leq \sum_{i=1}^{N_1}\left(p_i-p_{N_1+1}\right)
\end{equation}
and $N_2\in\mathbb{N}$ such that
\begin{equation}
\sum_{i=N_2+1}^{m}\left(p_{N_2}-p_i\right) < \epsilon \leq \sum_{i=N_2}^{m}\left(p_{N_2-1}-p_i\right).
\end{equation}
Here, we removed $\epsilon$ from the head of $\rho$, and distributed this probability mass over the tail of $\rho$. One can show see that when $ \varepsilon $ is larger, $ N_1$ becomes larger and $ N_2 $ becomes smaller; when $ \varepsilon = \delta(\rho,\id/m)$, the flattest state according to this construction will give us the maximally mixed state. For all $ \delta(\rho,\id/m)< \varepsilon \leq 1 $, the eigenvalues of the flattest state are simply given by 
\begin{equation}\label{key}
\tilde p_i = \frac{1}{d}, \qquad \forall i.
\end{equation} 
This state is majorized by all other states within the $\epsilon$-ball. \\

\subsection{General Hamiltonians: Construction of the flattest state}\label{app:B2flattest}
In this section, we turn to the case of general (finite-dimensional) Hamiltonians. We show that for any quantum state $\rho$, and for any smoothing parameter $ \epsilon $, the flattest state as defined in Def.~\ref{def:flattest} always exists.
\begin{theorem}\label{thm:flattest_exists}
Consider any $ d $-dimensional state $ \rho $ which is block-diagonal with respect to $ H $.	For any $\epsilon >0$, there exists a state $\flattest $ such that $ \flattest\in\bdepsball(\rho) $ and for any other state $ \rho'\in\epsball(\rho) $, $ \rho'\rightarrow\flattest$ is possible via thermal operations.
\end{theorem}
\begin{proof}		
	We begin by noting that it suffices to prove that any state $ \rho'\in\bdepsball(\rho) $ goes to $ \flattest $ via thermal operations. This is because if we have some $ \rho'' \in\epsball(\rho)$ that is not block-diagonal, we can nevertheless first apply a map $ \mathcal{M} $ that decoheres $ \rho'' $ in the energy eigenbasis. The resulting state $ \mathcal{M}(\rho'') $ is within $ \bdepsball(\rho) $, this is shown by invoking the data processing inequality for trace distance:
	\begin{equation}\label{key}
	\delta(\mathcal{M}(\rho),\mathcal{M}(\rho'')) \leq \delta(\rho,\rho'') \leq \epsilon.
	\end{equation}
	We continue by denoting $p=\lbrace p_i\rbrace_i$ as the $\beta$-ordered eigenvalues of $\rho$ with corresponding energy levels $E=\lbrace E_i\rbrace_i$. To prove this theorem, we provide an explicit method to construct $ \flattest $ for any $ \varepsilon $ such that any other state in $ \bdepsball(\rho) $ will thermo-majorize $ \flattest $. 
	
	We will consider two cases. If $\epsilon$ is large enough, such that
	\begin{equation}
	\delta(\rho,\tau_\beta) = \frac{1}{2}\sum_{i=1}^{n} \left|p_i - \frac{e^{-\beta E_i}}{\sum_{j=1}^{n}e^{-\beta E_j}}\right| \leq \epsilon,
	\end{equation}
	then this means the thermal state $ \tau_\beta \in\bdepsball(\rho)$. Since we know all block-diagonal states thermo-majorize $ \tau_\beta $, by setting $ \flattest = \tau_\beta $ we have that for any $\epsilon\geq \delta(\rho,\tau_\beta)$, the flattest state clearly exists.
	
	For the case where $\epsilon\leq \delta(\rho,\tau_\beta)$, it is not as straightforward to see that the flattest state exists. However, we will present a way to construct this state, and prove that this state is thermo-majorized by all other states within the $\epsilon$-ball. For any $ \epsilon >0 $, we perform the following steps to construct a state $ \hat{\rho} $, which later we show that $ \hat{\rho} = \flattest $:\\[5pt]
	\textbf{Step 1: Determine an integer $ \mathbf{M}$, and partially decrease the first $ \mathbf{M}$ ($ \mathbf{\beta} $-ordered) eigenvalues $ \mathbf{p_1, \cdots,p_M} $. }Define the function 
	\begin{equation}\label{eq:fm}
	F(m) = \sum_{i=1}^m p_i - p_{m+1} e^{\beta E_{m+1}}\sum_{i=1}^{m} e^{-\beta E_i}, \quad m\in\lbrace 1,d-1\rbrace.
	\end{equation}
	Note that due to the fact that $ p_i $ are $ \beta $-ordered, $ F(1)\geq 0 $, $ F(d-1) \geq \varepsilon $, and this function is non-decreasing with respect to $ m $ (Lemma \ref{lem:Fm_properties}, Appendix \ref{app:D_TL}).
	Therefore, we may find the smallest integer $ 1 \leq M \leq d-1 $ such that
	\begin{equation}\label{eq:fmcondition}
	\epsilon \leq F(M).
	\end{equation}
	This value $ M $ is the number of eigenvalues we cut from $ \rho $ to obtain $ \hat{\rho} $. Firstly, denote the total probability mass of these eigenvalues as
	\begin{equation}\label{key}
	A(M) = \sum_{i=1}^M p_i,
	\end{equation}
	and note that since $ \varepsilon \leq F(M) $, $ \varepsilon < A(M) $ is also true.
	We now denote the eigenvalues of $ \hat{\rho} $ as $ \hat p $, and for $ i\leq M $, let
	\begin{equation}\label{eq:flattesteigenfor1toM}
	\hat p_i = \frac{A(M)-\varepsilon}{\sum_{i=1}^M e^{-\beta E_i}} \cdot e^{-\beta E_i}.
	\end{equation}
	
	From this construction in Eq.~\eqref{eq:flattesteigenfor1toM} we see that
	\begin{equation}\label{key}
	\sum_{i=1}^M \hat p_i = A(M)-\varepsilon,
	\end{equation}
	such that a total amount of exactly $ \varepsilon $ is cut from $ p_1,\cdots,p_M $ to obtain $ \hat p_1, \cdots, \hat p_M $. Furthermore, the first $ M $ eigenvalues are cut in a way such that they have the same ``advantage'' in $ \beta $-ordering, i.e.
	\begin{equation}\label{eq:betaPreservance}
		\hat{p}_1e^{\beta E_1} = \dotsc = \hat{p}_{M}e^{\beta E_{M}} \geq \hat{p}_{M+1}e^{\beta E_{M+1}}.
	\end{equation} 
	The inequality follows from our choice of $M$ as described by Eqns. \eqref{eq:fm} and \eqref{eq:fmcondition}. Firstly, $p_1, \cdots,p_M$ have the same beta-ordering by construction, therefore the beta-ordering can differ from the initial state only by one way, i.e. by reducing the first $M$ eigenvalues such that $\hat{p}_ie^{\beta E_i} < \hat{p}_{M+1}e^{\beta E_{M+1}}$ for all $i \leq M$. However, if this is true, then Eq.~\eqref{eq:fmcondition} requires that more than $\epsilon$ would have to be cut from $ p_1,\cdots,p_M $. Since this is not the case, $\beta$-ordering is preserved. 
	\\[5pt]
	\textbf{Step 2: Adding $ \mathbf{\varepsilon} $ onto the eigenvalues $ \mathbf{p_N,\cdots,p_d }$ for some integer $ \mathbf{N\geq M} $ to renormalize.}\\
	In a similar way, we can also determine another integer $M \leq N < d $ (the lower bound on $ N $ holds whenever the trace distance $ \delta(\rho,\tau)\leq\varepsilon $), which tells us how many eigenvalues we have to increase. For any integer $ 2 \leq m \leq d $, consider the function
	\begin{equation}\label{key}
	G(m) = p_{m-1} e^{\beta E_{m-1}} \sum_{i=m}^d e^{-\beta E_i} - \sum_{i=m}^d p_i.
	\end{equation}
	Note that by Lemma \ref{lem:Gm_properties} (Appendix \ref{app:D_TL}), $ G(d) \geq 0 $, $ G(2) \geq \varepsilon $ and $ G(m) $ is non-increasing in $ m \in\lbrace 2,d\rbrace $.
	Let $N$ be the largest integer such that
	\begin{equation}\label{eq:gncond}
	\varepsilon \leq G(N).
	\end{equation}
	Once $ N $ is determined, denote the total probability mass 
	\begin{equation}
	B(N) = \sum_{i=N}^d p_i .
	\end{equation}
	We proceed to increase the probabilities $ p_N,\cdots,p_d $ in the following way to obtain $ \hat p_N,\cdots,\hat p_d $: for $ N \leq i\leq d $, let
	\begin{equation}\label{eq:hatp_eigen_Ntod}
	\hat p_i = \frac{B(N)+\varepsilon}{\sum_{i=N}^d e^{-\beta E_i}}\cdot e^{-\beta E_i}.
	\end{equation}
	Note that due to this construction, these eigenvalues are increased so that they again have the same $ \beta $-ordering advantage: $ \hat p_{N+1} e^{\beta E_{N+1}} \geq \hat p_N e^{\beta E_N} = \cdots = \hat p_d e^{\beta E_d}$. The inequality follows from our choice of $N$ in a similar way to the inequality of Eq. (\ref{eq:betaPreservance}). Eq.~\eqref{eq:gncond} ensures that more than $\varepsilon$ has to be added to the eigenvalues to change the $\beta$-ordering.\\[5pt]
	\textbf{Step 3: Keep all the other eigenvalues.}\\
	The last step in defining $ \hat{\rho} $ is such that for all $ M<i<N $, the eigenvalues are left untouched, i.e. $ \hat p_i = p_i $.\\

	We have now finished the task of constructing a particular flat state $ \hat{\rho}$, which is diagonal in the same basis as $ \rho $, with eigenvalues denoted by $ \hat p $. Now, what remains is to show that $ \hat{\rho} $ is thermo-majorized by all states $ \rho'\in\bdepsball(\rho) $, and therefore $ \hat{\rho} = \flattest $. To do this, we will divide the thermo-majorization curve up into three different regions, similar to what we did earlier. These regions are depicted in Fig.~\ref{flat1}.
	
	\begin{figure}[h!]
		\centering\includegraphics[width=\linewidth]{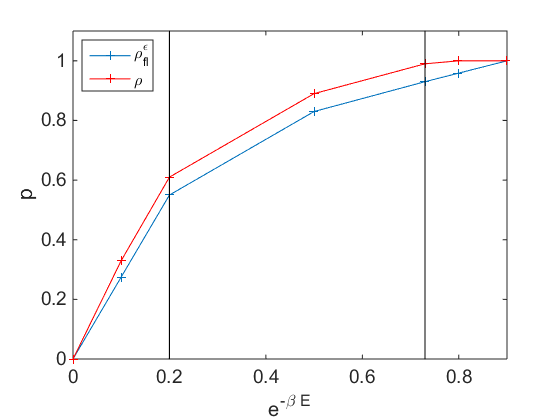}
		\caption{The thermo-majorization diagram of the flattest state divided up into three regions. In this particular example, $M=2$ and $N=5$. This means the first two beta-ordered eigenvalues are cut (by a total amount of $ \varepsilon $), while from the fifth eigenvalue onwards, each eigenvalue is increased. In the middle zone, the eigenvalues are unchanged.}
		\label{flat1}
	\end{figure}
	
	Firstly, let us consider the region $ x\in \left[0,\sum_{i=1}^{M}e^{-\beta E_i}\right] $.
	Since we have seen that $ \hat p_1, \cdots, \hat p_M $ have the same $ \beta $-ordering advantage, the thermo-majorization curve $ c_{\hat{\rho}} $ is a straight line within this interval. 
	Furthermore, if we compare the curves $ c_\rho,c_{\hat{\rho}} $ at the rightmost end of the interval, i.e. $ x_M = \sum_{i=1}^{M}e^{-\beta E_i}$, we see that
	\begin{equation}\label{key}
	 c_{\rho} (x_M) = c_{\hat{\rho}} (x_M) + \varepsilon.
	\end{equation}
	This means that $ c_{\hat{\rho}} $ has a thermo-majorization curve that achieves the lower bound given in Theorem \ref{theorem:epsborder}.
	Now, is it possible for another state $ \rho' $ to have a thermo-majorization curve $ c_{\rho'} < c_{\hat{\rho}} $ at any point in this interval? Since we know that thermo-majorization diagrams are concave, it follows that if such a curve exists, then $ c_{\rho'} (x_M) < c_{\hat{\rho}}(x_M) $ has to hold as well. 
	However, by Theorem \ref{theorem:epsborder} this is impossible, and we arrive at a contradiction. This implies that for any $ \rho'\in\bdepsball (\rho) $, in the interval $ x\in \left[0,\sum_{i=1}^{M}e^{-\beta E_i}\right]$, we always have $ c_{\rho'} \geq c_{\hat{\rho}} $.
	
	The second region we consider is the interval $ x\in \left[\sum_{i=1}^{M}e^{-\beta E_i},\sum_{i=1}^{N}e^{-\beta E_i}\right] $. 
	For this entire region, we have that
	$ c_{\rho} = c_{\hat{\rho}} + \varepsilon$.
	Therefore, by the same reasoning, any $ \rho' $ satisfies $ c_{\rho'} \geq c_{\hat{\rho}} $ in this region.
	
	Finally, we see that the same reasoning applies to the third interval $ x\in\left[ \sum_{i=1}^{N}e^{-\beta E_i}, Z \right] $. Recall that at $ x_N = \sum_{i=1}^{N}e^{-\beta E_i} $, we have $ c_\rho (x_N) = c_{\hat
	\rho} (x_N)+\epsilon$, and within this interval $ c_{\hat{\rho}} $ is again a straight line. For any other $ c_{\rho'} $, since it is concave, if $ c_{\rho'} < c_{\hat{\rho}}$ within this interval, then $ c_{\rho'} (x_N) < c_{\hat{\rho}} (x_N) $ as well, which again leads to a contradiction.
	
	
	Note that the thermo-majorization diagram of any other state $ \rho'\in\bdepsball(\rho) $ lies within these three regions, if the Hamiltonian stays invariant. Combining our analysis for the three regions, we have shown that any such $\rho'$ will have a thermo-majorization curve $ c_{\rho'} \geq c_{\hat{\rho}}$ at all points of the diagram. In other words, given any state $ \rho'\in\epsball_D (\rho) $, $ \rho' $ always thermo-majorizes $ \hat{\rho} $. Therefore, by definition, $ \hat{\rho} = \flattest $. 
\end{proof}

\subsection{General Hamiltonians: steepest state}\label{app:B3steepest}

In this section, we give our results on the steepest state. We first show that there does not, in general, exist a steepest state. Then, we present a way to construct the steepest state for small $\epsilon$. Finally, we use this steepest state to define our particular steep state.	

\subsubsection{Non-existence of a general steepest state}


To show that there is no steepest state wrt TO, it suffices to show that there is no steepest state wrt CTO. This can be seen as follows: if there is no steepest state wrt CTO, it means that for any candidate state $ \bar\rho_{\rm steep}^\varepsilon $ chosen, there exists at least one other state $  \check\rho_{\rm steep}^\varepsilon $ where $  \bar\rho_{\rm steep}^\varepsilon \rightarrow  \check\rho_{\rm steep}^\varepsilon $ is not possible via CTO. If $  \bar\rho_{\rm steep}^\varepsilon \rightarrow  \check\rho_{\rm steep}^\varepsilon $ is not possible via CTO, it is also not possible via TO. Therefore, by the same definition, there exists no steepest state wrt TO. 

Consider the block diagonal state $\rho$, with eigenvalues $\{p_i\}_i = \{0.55,0.35,0.1\}$ and corresponding $\beta$-factors $\{e^{\beta E_i}\}_i = \{1,2,8\}$. Denote the eigenvalues of the thermal state $ \tau $ as $ \{q_i\}_i $. Consider all states within $\epsball_D(\rho)$ for $\epsilon = 0.45$. Since a steepest state maximizes the R\'enyi divergences for all $\alpha\in\mathbb{R}$, we know that in particular 
\begin{eqnarray}
D_0(\steepest||\tau) &=& \max_{\hat{\rho}\in\epsball(\rho)}D_0(\hat{\rho}||\tau) \nonumber\\
&=& \max_{\hat{\rho}\in\epsball(\rho)}(-\text{log}\sum_{i:\hat{p}_i>0}q_i) \nonumber\\
&=& \min_{\hat{\rho}\in\epsball(\rho)}(\text{log}\sum_{i:\hat{p}_i>0}q_i).
\end{eqnarray}
Thus, in order for a state $\hat{\rho}$ to be steepest, it has to minimize $q_i$ for which $p_i$ are nonzero. Note that $q_i$ are inversely proportional to the $\beta$-factors of $\rho$. Thus, in order to obtain the steepest state, we have to cut the eigenvalues that correspond to large $\beta$-factors. In our example, this means we would like to cut the $0.55$ eigenvalue. We cannot do this, however, because the resulting state would no longer be within $\epsball(\rho)$. Thus, we have to cut the other two eigenvalues to attain the maximum of the divergences for $\alpha=0$. We define the eigenvalues of $\hat{\rho}$ by $\{\hat{p}_i\}_i = \{1,0,0\}$.

Note that a steepest state has to maximize $D_\alpha(\hat{\rho}||\tau)$ for all values of $\alpha$. Thus, if we can find an $\alpha$ for which the state that we just constructed does not maximize the R\'enyi divergence, then we have proved that no steepest state exists at all, for this scenario. In particular, if we can find such $\alpha\in(0,1]$, then this also shows that the new smoothed divergences and the smoothed R\'enyi divergences may be different, since this would imply that a single state cannot always attain the maximum for all $\alpha\in[0,1]$.

Consider the block diagonal state $\tilde{\rho}\in\epsball_D(\rho)$, with eigenvalues given by $\{\tilde{p}_i\}_i = \{0.45, 0, 0.55\}$, corresponding to the same $\beta$-factors as before. For this state, we find that for $\alpha=1$,
\begin{eqnarray}
D_1(\tilde{\rho}||\tau) &=& \sum_{i=1}^3 \tilde{p}_i \text{log}\frac{\tilde{p}_i}{\tilde{q}_i} \nonumber\\
&=& 0.45\text{log}(\frac{0.45\cdot11}{8}) + 0.55\text{log}(\frac{0.55}{8}) \nonumber\\
&>& \text{log}(\frac{11}{8}) =  \sum_{i=1}^3 \hat{p}_i \text{log}\frac{\hat{p}_i}{\hat{q}_i} \nonumber\\
&=& D_1(\hat{\rho}||\tau).
\end{eqnarray}
Thus, $\hat{\rho}$ does not maximize the R\'enyi divergence for $\alpha=1$. Since for this case, $\hat{\rho}$ was the unique state maximizing the R\'enyi divergence for $\alpha=0$, there exists no steepest state within $\epsball_D(\rho)$.\\


\subsubsection{The steepest state for small $\epsilon$}
\begin{theorem}
	Consider any $d$-dimensional state $\rho$ which is block-diagonal with respect to $  H $, and let the $\beta$-ordered eigenvalues of $\rho$ be given by $\{p_i\}_i$, with corresponding energy levels $E_i$. Then, if $\epsilon$ is bounded such that $ \varepsilon \leq \min\lbrace\varepsilon_A,\varepsilon_B,\varepsilon_C\rbrace $, where
	\begin{eqnarray} 
    \varepsilon_A &:=& \min_{i: p_i>0}p_i,\label{eq:epsbounds_A}\\
    \varepsilon_B &:=& \min_{i:E_i>E_1}\left(\frac{p_1e^{\beta E_1}-p_ie^{\beta E_i}}{e^{\beta E_i}-e^{\beta E_1}}\right),\label{eq:epsbounds_B}\\
    \varepsilon_C &:=& \min_{i: p_i>0, E_i>E_k}\left(\frac{p_ie^{\beta E_i}-p_ke^{\beta E_k}}{e^{\beta E_i}-e^{\beta E_k}}\right),\label{eq:epsbounds_C}
	\end{eqnarray}
	then a steepest state $\steepest$ as defined in Def.~\ref{def:steepest} exists, and its eigenvalues are given by
	\begin{equation} \label{eq:stip}
	\hat{p}_i = \begin{cases}
	p_i+\epsilon & \text{if } i=1\\
	p_i-\epsilon & \text{if } i=k\\
	p_i & \text{otherwise},
	\end{cases}
	\end{equation}
	where $k$ is the largest index for which $p_k>0$.
\end{theorem}
\begin{proof}
	We only have to show that the state $\steepest$ that we defined in Eq.~\eqref{eq:stip}, is indeed the steepest state. Thus, we want to show that for any other state $ \rho'\in\epsball_D(\rho) $, $ \steepest\rightarrow\rho' $ is possible via thermal operations. We will do this by comparing the thermo-majorization curves $c_{\steepest}$ and $c_{\rho'}$ of $\steepest$ and $\rho'$ respectively. We will divide $c_{\steepest}$ up into four different regions, just like we did before, and show for each region that $c_{\rho'}\leq c_{\steepest}$. These regions are depicted in Fig.~\ref{stip1}.		
	\begin{figure}[H]
		\centering\includegraphics[width=\linewidth]{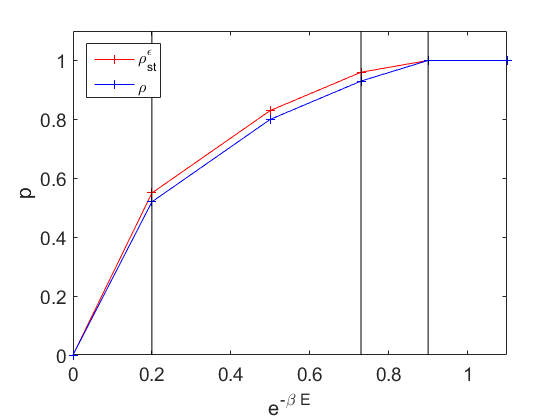}
		\caption{The thermo-majorization diagram divided up into four different regions. In this example, $k = 4$.}
		\label{stip1}
	\end{figure}
	
	Firstly, let us consider the region $ x\in \left[0,e^{-\beta E_1}\right] $. Because $c_{\steepest}$ in this entire region is a straight line, the only way to surpass it, is by having a steeper slope. For this to happen, the eigenvalues $\{p_i'\}_i$ of $\rho'$ must satisfy
	\begin{equation} \label{eq:reg1}
	p_i'e^{\beta E'_i} > \hat{p}_1e^{\beta E_1},
	\end{equation}
	for at least some $1\leq i\leq d$. We use the bound on $\epsilon$ given in Eq.~\ref{eq:epsbounds_A}-\eqref{eq:epsbounds_C} to show that this is impossible.	This bound consists of three parts, of which one is given by
	\begin{equation}
	\epsilon \leq \min_{i:E_i>E_1}\left(\frac{p_1e^{\beta E_1}-p_ie^{\beta E_i}}{e^{\beta E_i}-e^{\beta E_1}}\right).
	\end{equation}
	In particular, this bound implies that for all $1<i\leq d$ for which $E_i>E_1$,
	\begin{equation}
	\epsilon \leq \left(\frac{p_1e^{\beta E_1}-p_ie^{\beta E_i}}{e^{\beta E_i}-e^{\beta E_1}}\right).
	\end{equation}
	Rewriting this yields that for these $i$,
	\begin{equation}
	\epsilon\left(e^{\beta E_i}-e^{\beta E_1}\right) \leq \left(p_1e^{\beta E_1}-p_ie^{\beta E_i}\right).
	\end{equation}
	Note that this equation trivially holds if $E_i\leq E_1$. Thus, we find that for all $1<i\leq d$,
	\begin{equation}
	\hat{p}_1e^{\beta E_1} = \left(p_1+\epsilon\right)e^{\beta E_1} \geq \left(p_i+\epsilon \right)e^{\beta E_i} \geq p_i'e^{\beta E'_i},
	\end{equation}
	which means Eq~\ref{eq:reg1} does not hold. Thus, for this region we find that for any state $\rho'\in\epsball_D(\rho)$, $c_{\rho'}\leq c_{\steepest}$.	
	
	Next, we consider the interval $ x\in \left[e^{-\beta E_1},\sum_{i=1}^{k-1}e^{-\beta E_i}\right] $. Note that for all $x$ within this interval, 
	\begin{equation}
	c_{\steepest} (x) = c_{\rho} (x) + \varepsilon.
	\end{equation}
	This means that $ \steepest $ has a thermo-majorization curve that achieves the upper bound given in Theorem \ref{theorem:epsborder}. Thus, for this region we also find that for any state $\rho'\in\epsball_D(\rho)$, $c_{\rho'}\leq c_{\steepest}$.
	
	The third region we consider is the interval $ x\in \left[\sum_{i=1}^{k-1}e^{-\beta E_i},\sum_{i=1}^{k}e^{-\beta E_i}\right] $. Similar to the previous interval, we will use the bound on $\epsilon$ to show that $c_{\steepest}$ cannot be surpassed.
	
	Note that in this region, $c_{\steepest}$ is a straight line with one endpoint given by $(\sum_{i=1}^{k}e^{-\beta E_i},1)$. Because thermo-majorization curves cannot surpass $1$, the curve $c_{\rho'}$ can only lie above $c_{\steepest}$ if $\rho'$ has an eigenvalue such that 
	\begin{equation} \label{eq:reg3}
	p'_ie^{\beta E_i} < \hat{p}_ke^{\beta E_k}.
	\end{equation}
	There are two ways to construct such eigenvalues. Either we can increase some eigenvalue $p_i$ that was originally equal to $0$, or we can partially cut a nonzero eigenvalue. However, if we choose to do the former, then the line segment still has to be moved to the region that we are currently looking at. The only way to do this, is by decreasing another eigenvalue such that its slope is even flatter. Thus, in both cases, we have to decrease an eigenvalue such that Eq.~\eqref{eq:reg3} is satisfied.	We again use the bound on $\epsilon$ given in Eq.~\eqref{eq:epsbounds_A}-\eqref{eq:epsbounds_C} to show that this is not possible. One of the parts of the bound is given by
	\begin{equation}
	\epsilon \leq \min_{i: p_i>0, E_i>E_k}\left(\frac{p_ie^{\beta E_i}-p_ke^{\beta E_k}}{e^{\beta E_i}-e^{\beta E_k}}\right),
	\end{equation}
	which implies that for all $i$ for which $p_i>0$ and $E_i>E_k$,
	\begin{equation}
	\epsilon \leq \left(\frac{p_ie^{\beta E_i}-p_ke^{\beta E_k}}{e^{\beta E_i}-e^{\beta E_k}}\right).
	\end{equation}
	Rewriting this yields that for these $i$, 
	\begin{equation}
	\epsilon\left(e^{\beta E_i}-e^{\beta E_k}\right) \leq \left(p_ie^{\beta E_i}-p_ke^{\beta E_k}\right).
	\end{equation}
	Note that this equation trivially holds if $E_i\leq E_k$. Thus, we find that for all $1<i\leq d$ for which $p_i>0$,
	\begin{equation}
	\hat{p}_ke^{\beta E_k} = \left(p_k-\epsilon\right)e^{\beta E_k} \leq \left(p_i-\epsilon\right)e^{\beta E_i} \leq p_i'e^{\beta E_i}.
	\end{equation}
	This contradicts Eq.~\ref{eq:reg3}, and thus we have that in this region, for any state $\rho'\in\epsball_D(\rho)$, $c_{\rho'}\leq c_{\steepest}$.
	
	Finally, for the interval $ x\in \left[\sum_{i=1}^{k}e^{-\beta E_i},\sum_{i=1}^{d}e^{-\beta E_i}\right] $, we find that
	\begin{equation}
	c_{\steepest} (x) = 1,
	\end{equation}
	because $k$ is the largest index for which $p_k$ is nonzero. Clearly, because states are normalized, it is impossible for any thermo-majorization curve to surpass this.
	
	Since for all regions, the thermo-majorization curve of $\steepest$ cannot be surpassed, $\steepest$ thermo-majorizes all other states within the $\epsilon$-ball, and is therefore the steepest state.
\end{proof}

\subsection{Existence of Thermal Operation that achieves approximate state transition}\label{subsec:approxTO}
In our work, we apply smoothing procedures on two states: the initial state $ \rho $ as well as the final state $ \sigma $. The reason for this might not be intuitive: indeed one might be satisfied to reach the target state $ \sigma'\approx_\varepsilon\sigma $ approximately, however why can we assume that we start out in another initial state $ \rho'\approx_\varepsilon\rho $? The following lemma rigorously explains the physical justification for doing so: if $ \rho'\rightarrow\sigma' $ is achievable by a thermal operation $ \mathcal{N} $, then if one applies $ \mathcal{N} $ to the original initial state $ \rho $, the final state obtained is always in a $ 2\varepsilon $-ball of the state $ \sigma $.
\begin{Lemma}
	Consider any quantum states $ \rho_S,\rho_S',\sigma_S,\sigma_S' $ such that $ \rho_S'\in\mathcal{B}^{\varepsilon_1} (\rho_S) $ and $ \sigma_S'\in\mathcal{B}^{\varepsilon_2} (\sigma_S) $. Then for any quantum channel $ \mathcal{N} $ such that $ \mathcal{N}(\rho_S')= \sigma_S'$, we have
	\begin{equation}\label{key}
		\tilde{\rho}_S := \mathcal{N} (\rho_S) \in \mathcal{B}^{\varepsilon_1+\varepsilon_2} (\sigma_S).
	\end{equation}
\end{Lemma}
\begin{proof}
	By assumption of the lemma we have that $ \delta(\rho_S,\rho_S')\leq\varepsilon_1 $ and $ \delta(\sigma_S,\sigma_S')\leq\varepsilon_2 $. Furthermore, by the data processing inequality of trace distance, we have
	\begin{equation}\label{key}
	\delta(\tilde{\rho}_S,\sigma_S') = \delta(\mathcal{N}(\rho_S),\mathcal{N}(\rho_S')) \leq \delta(\rho_S,\rho_S')\leq\varepsilon_1.
	\end{equation}
	On the other hand, we know from the triangle inequality that 
	\begin{eqnarray*}
	\delta(\mathcal{N}(\rho_S),\sigma_S) &=& \delta(\tilde\rho_S,\sigma_S) \leq \delta(\tilde{\rho}_S,\sigma_S') + \delta(\sigma_S,\sigma_S') \\
	&\leq& \delta(\tilde{\rho}_S,\sigma_S') +\varepsilon_2 \leq \varepsilon_1+\varepsilon_2.
	\end{eqnarray*}
\end{proof}
%
%
\section{Asymptotic Equipartition Property (AEP)}\label{app:C_AEP}
In this appendix we prove that the new smoothed divergences defined in Eq.~\eqref{eq:newdiv} satisfy the asymptotic equipartition property (this is stated in Theorem \ref{thm:AEP} of the main text). 
By this, we mean that for any $ \alpha\geq 0 $, when we consider our smoothed divergences for any states $\rho $ and $ \sigma $, then
\begin{equation}
\lim_{\varepsilon\rightarrow 0} \lim_{n\rightarrow\infty} \frac{1}{n} \hat D_\alpha^\varepsilon (\rho^{\otimes n}\|\sigma^{\otimes n}) = D(\rho\|\sigma).
\end{equation}
Such a property cannot be satisfied by the unsmoothed R{\'e}nyi divergences $ D_\alpha $, since the exact quantities are additive under tensor product, and therefore for any positive integer $ n $, the quantity $ \frac{1}{n} \hat D_\alpha (\rho^{\otimes n}\|\sigma^{\otimes n}) = D_\alpha (\rho\|\sigma) \neq D(\rho\|\sigma) $ in general. However, the usual smoothed versions $ D_\alpha^\varepsilon (\rho\|\sigma) $, as defined in Eq.~\eqref{eq:olddiv}, do satisfy this property.
\subsection{A $ \delta $-typical subspace}
To prove the AEP for the quantities $ \hat D_\alpha (\rho\|\sigma) $, we first need to establish a technical lemma regarding the typical subspace of $ \rho^{\otimes n} $. This can be done by using Hoeffding's inequality~\cite{hoeffding1963probability}. This lemma shows that as $n$ grows large, most of the weight of the eigenvalues of $ \rho^{\otimes n} $ lie within such a typical subspace.
\begin{Lemma} \label{lemma:hoeffding}
	For any quantum state $ \rho $ block-diagonal with respect to its Hamiltonian $ H $, consider $ n $ copies, $ \rho^{\otimes n} $. Let $\{\tilde{p}_k\}_k$ be the $\beta$-ordered eigenvalues of $\rho^{\otimes n}$, and let $\{\tilde{q}_k\}_k$ be the eigenvalues of $\tau^{\otimes n}$ in the same ordering as $\rho^{\otimes n}$. 
	Then according to the probability distribution given by $ \{\tilde{p}_k\}_k$,
	we have that for any $\delta>0$,
	\begin{equation}
	\text{Pr}\left(2^{n\left[D(\rho||\tau)-\delta\right]} \leq \frac{\tilde{p}_k}{\tilde{q}_k} \leq 2^{n\left[D(\rho||\tau)+\delta\right]}\right) \geq 1 - 2e^{-2n\delta^2}.
	\end{equation} 
\end{Lemma}
\begin{proof}
	First of all, note that if $ \rho $ is block-diagonal with respect to $  H $, then it commutes with the thermal state $\tau$. Therefore, both $\rho$ and $\tau$ can be diagonalized in the same, ordered basis. Written in such a basis, let us denote the eigenvalues of $\rho$ by $\{p_i\}_i$, and the eigenvalues of $\tau$ by $\{q_i\}_i$, and let $d$ be the dimension of $\rho$. Furthermore, without loss of generality we can order this common basis such that it corresponds to the $\beta$-ordering of $\rho$, such that
	$ p_1 e^{\beta E_1} \geq \dotsc \geq p_d e^{\beta E_d} $. Since each eigenvalue of $ \tau $ is given by $ q_i = \frac{1}{Z} e^{-\beta E_i}$, it follows directly that
	$\frac{p_1}{q_1}\geq\dotsc\geq\frac{p_d}{q_d}$.
	
	Next, we will introduce Hoeffding's inequality. Consider the sequence $X_1,\dotsc,X_n$ of independent and identically distributed random variables, where each random variable $X_j$ can assume the values $\{\text{log}\frac{p_i}{q_i}\}_i$ according to the probability distribution $\{p_i\}_i$. We denote the average of this sequence by $\overline{X}_n = \frac{1}{n}\sum_{j=1}^{n}X_j$, and the expected value by $\mu$. Then, by Hoeffding's inequality we have that for any $\delta>0$,
	\begin{equation}
	\text{Pr}\left(\left|\overline{X}_n - \mu\right|\geq \delta \right) \leq 2e^{-2n\delta^2}.
	\end{equation}
	Substituting the average and expected value gives us
	\begin{equation}
	\text{Pr}\left[\left|\frac{1}{n}\sum_{j=1}^{n}X_j- \sum_{i=1}^{d}p_i\text{log}\frac{p_i}{q_i}\right|\geq \delta \right] \leq 2e^{-2n\delta^2}.
	\end{equation}
	We will denote the value of $X_j$ by $\text{log}\frac{p_{F(j)}}{q_{F(j)}}$, where for each $ j\in\lbrace 1,\dotsc,n\rbrace $, the quantity $F(j)$ is a random variable across the alphabet $ \lbrace 1,\dotsc,d\rbrace $, according to the probability distribution given by $ \lbrace p_i\rbrace_{i=1}^d $. This yields
	\begin{equation}
	\text{Pr}\left[~\left|\frac{1}{n}\sum_{j=1}^{n}\text{log}\frac{p_{F(j)}}{q_{F(j)}} - \sum_{i=1}^{d}p_i\text{log}\frac{p_i}{q_i}\right|\geq \delta~ \right] \leq 2e^{-2n\delta^2}.
	\end{equation}
	Notice that $\sum_{i=1}^{d}p_i\text{log}\frac{p_i}{q_i} = D(\rho||\tau)$. Therefore, this is equivalent with
	\begin{equation}
	\text{Pr}\left[~\left|\frac{1}{n}\sum_{j=1}^{n}\text{log}\frac{p_{F(j)}}{q_{F(j)}} - D(\rho||\tau)\right|\geq \delta~ \right] \leq 2e^{-2n\delta^2}.
	\end{equation}
	Multiplying the equation within the large bracket by $n$, and taking the complement yields
	\begin{equation}
	\text{Pr}\left[~\left|\sum_{j=1}^{n}\text{log}\frac{p_{F(j)}}{q_{F(j)}} - nD(\rho||\tau)\right|\leq n\delta ~\right] \geq 1 - 2e^{-2n\delta^2}.
	\end{equation}
	If we now rewrite the sum of logarithms into a single logarithm, we get
	\begin{eqnarray*}
&\text{Pr}&\left\lbrace-n\delta \leq \left[\text{log}\prod_{j=1}^{n}\frac{p_{F(j)}}{q_{F(j)}} - nD(\rho||\tau)\right]\leq n\delta \right\rbrace \\ &\geq& 1 - 2e^{-2n\delta^2}.
	\end{eqnarray*}
	Finally, adding $nD(\rho||\tau)$ to the equation and exponentiating gives us
	\begin{eqnarray}\label{eq:typicalsubspace}
	&\text{Pr}&\left[2^{n\left(D(\rho||\tau)-\delta\right)} \leq \prod_{j=1}^{n}\frac{p_{F(j)}}{q_{F(j)}} \leq 2^{n\left(D(\rho||\tau)+\delta\right)} \right]\nonumber\\
	 &\geq& 1 - 2e^{-2n\delta^2}.
	\end{eqnarray}	
	The products $\prod_{j=1}^{n}p_{F(j)}$ and $\prod_{j=1}^{n}q_{F(j)}$, for any possible values of $ F(j) $ (There are $ d^n $ such different eigenvalues) are precisely eigenvalues of $\rho^{\otimes n}$ and $\tau^{\otimes n}$. This means that the desired inequality holds. For most of the probability mass of $ \tilde p_k$ for $k\in \lbrace 1,d^n\rbrace $, the value of $ \frac{\tilde p_k}{\tilde q_k} $ lies within the interval given in Eq.~\eqref{eq:typicalsubspace}.
\end{proof}

\subsection{Proof of Theorem \ref{thm:AEP}}\label{subsec:proofT2}
	For all $\alpha\geq 0$, we will try to find functions $f=f(n,\epsilon,\rho,\tau)$ and $g=g(n,\epsilon,\rho,\tau)$, such that
	\begin{equation}
		D(\rho||\tau)-f\leq \frac{1}{n}\hat{D}_\alpha^\epsilon(\rho^{\otimes n}||\tau^{\otimes n}) \leq D(\rho||\tau)+g,
	\end{equation}
	with these functions converging to $0$ as $n$ grows large and $\epsilon$ becomes small \footnote{For notational convenience, we drop the explicit dependence of $ f $ and $ g $ on the variables $ n,\varepsilon,\rho $ and $ \tau $ for the moment; but it will be helpful for the reader to take note that these functions will later vanish for all possible $ \rho,\tau $, in the limit $ \varepsilon\rightarrow 0 $ and $ n\rightarrow\infty $. }. It will become clear later that these functions do not converge to $0$ for all values of $\alpha$, if we fixate either $\epsilon$ or $n$. Since the smoothing procedure is different for regimes $ \alpha\in [0,1] $ and $ \alpha >1 $, we shall split the analysis into two different parts.
	
	We first consider the region $0\leq\alpha\leq1$. For $\alpha=0$, our new smoothed divergence is equal to the R\'enyi divergence of the steep state. Let us denote the eigenvalues of the steep state by $\{\hat{p}_k\}_k$ for $k=1,\dotsc,d^n$. By the definition of $ \hat D_\alpha $ in Eq.~\eqref{eq:newdiv}, we have that
	\begin{align}
	\frac{1}{n}\hat{D}_0^\epsilon(\rho^{\otimes n}||\tau^{\otimes n}) & = \frac{1}{n}D_0((\rho^{\otimes n})_{\rm steep}^\varepsilon||\tau^{\otimes n})\nonumber\\
	& = -\frac{1}{n}\text{log}\sum_{k:\hat{p}_k>0}\tilde{q}_k. \label{eq:lim0}
	\end{align}	
	For any given $ \varepsilon $, we obtain $\hat{p}_k$ from $\tilde{p}_k$ (defined in Lemma \ref{lemma:hoeffding}) by cutting off all the eigenvalues for which the ratio $\frac{\tilde{p}_k}{\tilde{q}_k} \leq \gamma$, where $ \gamma $ is a real-valued parameter that depends on $ \varepsilon $. In particular, if $\gamma\geq L = 2^{n\left[D(\rho||\tau)-\delta\right]}$, with $\delta = \sqrt{\frac{1}{2n}\ln\left(\frac{2}{\varepsilon}\right)}$, then 
	by Lemma~\ref{lemma:hoeffding}, we have that
	\begin{equation}\label{eq:hoeff}
	\sum_{k:\frac{\tilde{p}_k}{\tilde{q}_k}<L}\tilde{p}_k \leq 2e^{-2n\delta^2} = \varepsilon,
	\end{equation}	
	and therefore $ \hat p $ is $ \varepsilon $-close to $ \tilde p $.
	This means we can cut into the $ \delta $-typical region for $\frac{\tilde{p}_k}{\tilde{q}_k}$. Now we will use the fact that $\gamma\geq L$ to lower bound the right hand side of Eq.~\eqref{eq:lim0}. Since we cut at least all eigenvalues for which $\frac{\tilde{p}_k}{\tilde{q}_k}<L$, we have that for all $k$ for which $\hat{p}_k>0$, $\frac{\tilde{p}_k}{\tilde{q}_k}\geq L$. This means that $\tilde{q}_k\leq\frac{\tilde{p}_k}{L}$. Thus, 
	\begin{align*}
	-\frac{1}{n}\log \sum_{k:\hat{p}_k>0}\tilde{q}_k & \geq  -\frac{1}{n}\log \sum_{k:\hat{p}_k>0}\frac{\tilde{p}_k}{L}\\
	& \geq \frac{1}{n}\log L\\
	& = D(\rho||\tau)-\sqrt{\frac{1}{2n}\ln\left(\frac{2}{\varepsilon}\right)}.
	\end{align*}
	Note that when $ n\rightarrow\infty $, this bound converges to $ D(\rho||\tau) $ even for a finite value of $\varepsilon >0$.
	
	Next, we will give an upper bound for $\frac{1}{n}\hat{D}^\varepsilon(\rho^{\otimes n}||\tau^{\otimes n})$. Since the steepest state cuts away a probability mass of $\varepsilon$, if we denote the value $U = 2^{n(D(\rho||\tau)+\delta)}$, then we may write 
	\begin{align}
	&\frac{1}{n}\hat{D}^\epsilon(\rho^{\otimes n}||\tau^{\otimes n}) \\
	& = \frac{1}{n}D((\rho^{\otimes n})_{\rm steep}^\varepsilon||\tau^{\otimes n})\nonumber\\ 
	& = \frac{1}{n}\sum_{i:\frac{\tilde{p}_i}{\tilde{q}_i}>U} \hat{p}_i\log\frac{\hat{p}_i}{\tilde{q}_i} + \frac{1}{n}\sum_{i:\frac{\tilde{p}_i}{\tilde{q}_i}\leq U} \hat{p}_i\log\frac{\hat{p}_i}{\tilde{q}_i}\nonumber\\
	& \leq \frac{1}{n}\hat{p}_1\log{\frac{\hat{p}_1}{\tilde{q}_1}} + \frac{1}{n}\sum_{i>1:\frac{\tilde{p}_i}{\tilde{q}_i}>U} \tilde{p}_i\log\frac{\tilde{p}_1}{\tilde{q}_1} + \frac{1}{n}\sum_{i:\frac{\tilde{p}_i}{\tilde{q}_i}\leq U} \tilde{p}_iU\nonumber\\
	& \leq \frac{1}{n}(\tilde{p}_1+\varepsilon)\log{\frac{\tilde{p}_1+\varepsilon}{\tilde{q}_1}} +\frac{2}{n}e^{-n\delta^2}\log\frac{\tilde{p}_1}{\tilde{q}_1} + \frac{1-\varepsilon}{n}U\nonumber\\
	& = (1-\varepsilon)[D(\rho||\tau)+\delta] +  \frac{1}{n}(\tilde{p}_1+\varepsilon)\log{\frac{\tilde{p}_1+\varepsilon}{\tilde{q}_1}} +2e^{-n\delta^2}\log\frac{p_1}{q_1}\nonumber\\
	&=: D(\rho\|\tau) + g_1(n,\varepsilon,\rho,\tau).\label{eq:uppbound_steep}
	\end{align}
	For clarity, let us first write out 
	\begin{align}\label{eq:g1}
	g_1(n,\varepsilon,\rho,\tau):=& (1-\varepsilon)\delta - \varepsilon D(\rho\|\tau) + \frac{1}{n}(\tilde{p}_1+\varepsilon)\log{\frac{\tilde{p}_1+\varepsilon}{\tilde{q}_1}}\nonumber\\ &+2e^{-n\delta^2}\log\frac{p_1}{q_1}
	\end{align}
	Let us observe the terms left in Eq.~\eqref{eq:g1}, in the limit when $ \varepsilon\rightarrow 0$ and $ n\rightarrow\infty $, furthermore in a way such that $ \delta $ as defined in Eq.~\eqref{eq:hoeff} goes to zero as well (for example, one may take $ \varepsilon = n^{-1} $). Since $ D(\rho\|\tau) $ is upper bounded by $ \log d $, the first two terms will vanish in this limit. Next, note that $ \tilde{p_1} = p_1^n $ and $ \tilde{q_1} = q_1^n $, where $ p_1,q_1 $ are simply the eigenvalues of $ \rho,\tau $ that maximize $ \beta $-ordering. Therefore, the third term vanishes as long as $ \frac{\varepsilon \log\varepsilon}{n} \rightarrow 0$, which is true whenever $ \delta\rightarrow 0 $. Lastly, note that $ 2e^{-n\delta^2} = \sqrt{2\varepsilon} $, and since $ \frac{p_1}{q_1} $ is just a constant where $ q_1 >0 $ (the thermal state has full rank), the last term vanishes as well. This implies that $ g_1(n,\varepsilon,\rho,\tau)\rightarrow 0 $ for all $ \rho,\tau $. 
	
	By using the fact that the modified smoothed divergences in this region are given by the R\'enyi divergence of a single state, we can apply these bounds to the entire region $0<\alpha\leq1$; the R\'enyi divergences are monotonic in $\alpha$, such that
	\begin{align}\label{eq:bounds0to1}
	D(\rho||\tau) - \delta &\leq \frac{1}{n}\hat{D}_0^\epsilon(\rho^{\otimes n}||\tau^{\otimes n})\nonumber\\
	&\leq \frac{1}{n}\hat{D}_\alpha^\epsilon(\rho^{\otimes n}||\tau^{\otimes n})\nonumber\\
	&\leq \frac{1}{n}\hat{D}^\epsilon(\rho^{\otimes n}||\tau^{\otimes n})\nonumber\\
	&\leq D(\rho||\tau) + g_1(n,\epsilon,\rho,\tau),
	\end{align}
	which concludes the proof for this regime of $ \alpha\in [0,1] $.
	
	Next, we consider the region $\alpha >1$. In this region, our divergences are smoothed towards the flattest state,
	\begin{equation}
	\hat{D}_\alpha^\epsilon(\rho^{\otimes n}||\tau^{\otimes n}) = D_\alpha((\rho^{\otimes n})_{\text{fl}}^\varepsilon||\tau^{\otimes n}).
	\end{equation}
	We start by looking at an upper bound of $\frac{1}{n}\hat{D}_\infty^\epsilon(\rho^{\otimes n}||\tau^{\otimes n})$. Note that $ D_\infty (\rho\|\sigma) := \lim_{\alpha\rightarrow\infty} D_\alpha (\rho\|\sigma) $, and for $ \rho,\sigma $ that commute and have ordered eigenvalues $ \{p_i\}$ and $ \{q_i\} $ respectively, this quantity has a simplified expression:
	\begin{equation}\label{key}
	D_\infty (\rho\|\sigma) = \log \max_i \frac{p_i}{q_i}.
	\end{equation}
	Note that for the flattest state, given some $ \varepsilon >0 $, one can obtain the flattest state with eigenvalues $ \lbrace \hat p_i\rbrace_i $, which has a distance exactly $ \varepsilon $-close to $ \rho^{\otimes n} $. Let $ L $ be the real-valued parameter, such that all eigenvalues for which $\frac{\tilde{p}_i}{\tilde{q}_i} \geq L$ are partially decreased, to obtain the new values $\frac{\hat{p}_i}{\tilde{q}_i}= L$ instead. Therefore, $ L $ corresponds to the largest $ \beta $-ordering gradient for the flattest state. 
	
	One can upper bound $ L $ by using Hoeffding's inequality to conclude that for $\delta = \sqrt{\frac{1}{2n}\ln\left(\frac{2}{\epsilon}\right)}$, we have 
	\begin{align}
	\sum_{i:\frac{\tilde{p}_i}{\tilde{q}_i}\geq n\left[D(\rho||\tau)+\delta\right]} \tilde{p}_i \leq 2e^{-2n\delta^2} = \epsilon.
	\end{align}
	This means that one would be able to cut through all eigenvalues of $ \tilde p_i $ where $ \frac{\tilde p_i}{q_i} \geq n\left[D(\rho||\tau)+\delta\right] $, and therefore 
	\begin{equation}
	L\leq n\left[D(\rho||\tau)+\delta\right].
	\end{equation}
	
Thus, we can always obtain a new distribution $ \hat p $ such that $ \frac{\hat p_i}{\tilde q_i} \leq L $ holds for all eigenvalues $ \hat p_i $. For the eigenvalue of the flattest state which has largest $ \beta $-ordering, this yields
	\begin{align}
	\frac{1}{n}\hat{D}_\infty^\epsilon(\rho^{\otimes n}||\tau^{\otimes n}) &\leq \frac{1}{n} \log (L) \nonumber\\
	&= D(\rho||\tau)+\sqrt{\frac{1}{2n}\ln\left(\frac{2}{\epsilon}\right)}.
	\end{align}
	
	Next, we look for a lower bound for the case of $\alpha=1$. To do so, we need to analyze another quantity: denote $L'$ as the smallest $\beta$-value of the flattest state (therefore, $ L' \leq L $). Let us try to find a lower bound for $ L' $. This can be done by noting that, the total probability mass of the smallest $\beta$-factors will be larger than $\varepsilon$, since $\varepsilon$ is distributed across these eigenvalues. More precisely, if we consider the set $ S = \lbrace i|\frac{\hat p_i}{q_i} = L'\rbrace $, then $ {\rm Prob} (S) \geq \varepsilon $. Therefore,
	\begin{equation}
	\sum_{i\in S} L' q_i \geq \varepsilon \qquad\implies\qquad L' \geq \frac{\varepsilon}{\sum_{i\in S}q_i} \geq \varepsilon.
	\end{equation}
	
	Therefore, for $U = 2^{n(D(\rho||\tau)+\delta)}$ and $M = 2^{n(D(\rho||\tau)-\delta)}$, we have that
	\begin{align}\label{eq:lowbound_flat}
	&\frac{1}{n}\hat{D}^\epsilon(\rho^{\otimes n}||\tau^{\otimes n}) \nonumber\\
	&= \frac{1}{n}D((\rho^{\otimes n})_{\text{fl}}^\varepsilon||\tau^{\otimes n})\nonumber\\
	& = \frac{1}{n}\sum_i \hat{p}_i \log\frac{\hat{p}_i}{\tilde{q}_i}\nonumber\\		
	& \geq \frac{1}{n}\sum_{i: M\leq\frac{\tilde{p}_i}{\tilde{q}_i}\leq U} \hat{p}_i \log\frac{\hat{p}_i}{\tilde{q}_i} + \frac{1}{n}\sum_{i: \frac{\tilde{p}_i}{\tilde{q}_i}<M} \hat{p}_i \log L'\nonumber\\
	& \geq (1-2e^{-n\delta^2}-\varepsilon)\cdot [D(\rho||\tau)-\delta] - \frac{1}{n}\log \frac{1}{\varepsilon} \sum_{i: \frac{\tilde{p}_i}{\tilde{q}_i}<M} \hat{p}_i \nonumber\\
& \geq (1-2e^{-n\delta^2}-\varepsilon)\cdot [D(\rho||\tau)-\delta] - \frac{1}{n} \log \frac{1}{\varepsilon}\nonumber\\
	&=: D(\rho\|\tau) - g_2(n,\varepsilon,\rho,\tau),
	\end{align}
	where in the last inequality, since $ \log\frac{1}{\varepsilon} > 0 $, we can use the bound $\sum_{i: \frac{\tilde{p}_i}{\tilde{q}_i}<M} \hat{p}_i  \leq 1$. Let us again write out
	\begin{equation}\label{eq:g2}
	g_2(n,\varepsilon,\rho,\tau):= \delta + (2e^{-n\delta^2}+\varepsilon) [D(\rho\|\tau)-\delta]-\frac{\log\varepsilon}{n}.
	\end{equation}
	 Note that we are taking the limit $ \varepsilon\rightarrow 0 $ and $ n\rightarrow\infty $ such that $ \delta $ also vanishes. Since $ D(\rho\|\tau) $ is upper bounded, the second term also vanishes. Finally, the last term vanishes as long as $ \delta $ vanishes as well.
	
	Thus, for the regime $ \alpha >1 $, one may conclude that
	\begin{align}
	D(\rho\|\tau) - g_2(n,\varepsilon,\rho,\tau)
	&\leq \frac{1}{n}\hat{D}^\epsilon(\rho^{\otimes n}||\tau^{\otimes n})\nonumber\\
	&\leq \frac{1}{n}\hat{D}_\alpha^\epsilon(\rho^{\otimes n}||\tau^{\otimes n})\nonumber\\
	&\leq \frac{1}{n}\hat{D}_\infty^\epsilon(\rho^{\otimes n}||\tau^{\otimes n})\nonumber\\
	&\leq D(\rho||\tau) + \delta.
	\end{align}	

By combining all the bounds we proved here in Section \ref{subsec:proofT2}, one can also show that given finite values of $ n $ and $ \varepsilon $, it suffices to check only \emph{a single sufficient condition} (in contrast with a continuous family of inequalities) for the approximate state transition $ (\rho^{\otimes n})_{\rm steep}^\varepsilon \rightarrow (\sigma^{\otimes n})_{\rm fl}^\varepsilon $.

\begin{corollary}\label{cor:onecondition}
	Consider states $ \rho, \sigma $, and for any real number $ \beta >0 $ and a Hamiltonian $ H $, let $ \tau_\beta = \frac{1}{\tr(e^{- \beta H})}e^{- \beta H}$. Moreover, consider any positive integer $ n $ and any $ \varepsilon >0 $. If
	\begin{equation*}\label{key}
F(\rho,\tau_\beta)	\geq F(\sigma,\tau_\beta) + \beta^{-1}\Delta(n,\varepsilon,\rho,\sigma,\tau_\beta),
	\end{equation*}
	is satisfied, where 
	\begin{equation}\label{key}
	\Delta(n,\varepsilon,\rho,\sigma,\tau_\beta):=\delta + f(n,\varepsilon,\rho,\sigma,\tau_\beta),
	\end{equation}
	where the first term $ \delta = \sqrt{\frac{1}{2n}\ln\frac{2}{\varepsilon}} $ and the second term
	\begin{equation*}\label{key}
	f(n,\varepsilon,\rho,\sigma,\tau_\beta) := {\rm max} ~[g_1(n,\varepsilon,\sigma,\tau_\beta),g_2(n,\varepsilon,\rho,\tau_\beta)],
	\end{equation*}
	$ g_1(n,\varepsilon,\sigma,\tau_\beta),g_2(n,\varepsilon,\rho,\tau_\beta) $ defined in Eq.~\eqref{eq:g1} and~\eqref{eq:g2}, then the transition $ (\rho^{\otimes n})_{\rm steep}^\varepsilon \rightarrow (\sigma^{\otimes n})_{\rm fl}^\varepsilon $ is possible via Thermal Operations, with a bath being of inverse temperature $ \beta $.
\end{corollary}
\begin{proof}
	By Theorem \ref{thm:equiv}, if $\hat F_\alpha^\varepsilon(\rho^{\otimes n},\tau_\beta^{\otimes n}) \geq \hat F_\alpha^\varepsilon(\sigma^{\otimes n},\tau_\beta^{\otimes n}) $ for all $ \alpha\geq 0 $, then the transition $ (\rho^{\otimes n})_{\rm steep}^\varepsilon \rightarrow (\sigma^{\otimes n})_{\rm fl}^\varepsilon $ is possible. Taking the definition in Eq.~\eqref{eq:defgenf}, this translates to 
	\begin{equation}\label{eq:coro1}
	\beta^{-1} \frac{1}{n} \hat D_\alpha(\rho^{\otimes n}\|\tau_\beta^{\otimes n}) \geq \beta^{-1} \frac{1}{n} \hat D_\alpha(\sigma^{\otimes n}\|\tau_\beta^{\otimes n}).
	\end{equation} 
	Eq.~\eqref{eq:coro1} will be satisfied for all $ \alpha\geq 0 $ if it is satisfied for both regimes $ \alpha\in [0,1] $ and $ \alpha\in (1,\infty) $. Therefore, let us look at the first regime: by using the upper bound in Eq.~\eqref{eq:bounds0to1} for $ \sigma $ and corresponding the lower bound for $ \rho $, we have one condition:
	\begin{equation}\label{key}
	D(\rho\|\tau_\beta)-\delta \geq D(\sigma\|\tau_\beta)+ g_1 (n,\varepsilon,\rho,\tau_\beta)
	\end{equation}
	which is sufficient for Eq.~\eqref{eq:coro1}, and can be rewritten as 
	\begin{equation}\label{eq1}
	F(\rho,\tau_\beta)\geq F(\sigma,\tau_\beta)+ \beta^{-1} [\delta +g_1 (n,\varepsilon,\sigma,\tau_\beta)].
	\end{equation}
	Similarly, for the second regime $ \alpha\in(1,\infty) $ one can also find the sufficient condition
	\begin{equation}\label{eq2}
	F(\rho,\tau_\beta)\geq F(\sigma,\tau_\beta)+ \beta^{-1} [\delta +g_2 (n,\varepsilon,\rho,\tau_\beta)].
	\end{equation}
	Since we need both Eqns.~\eqref{eq1} and \eqref{eq2} to hold, taking the maximum between $ g_1 (n,\varepsilon,\sigma,\tau_\beta) $ and $ g_2 (n,\varepsilon,\rho,\tau_\beta) $ suffices. Moreover, let us recall that in the limit $ \varepsilon\rightarrow 0 $, $ n\rightarrow\infty $ such that $ \delta\rightarrow 0 $ as well, we have that both $ g_1 (n,\varepsilon,\sigma,\tau_\beta) $ and $ g_2 (n,\varepsilon,\rho,\tau_\beta) $ vanish, hence recovering $ F(\rho,\tau_\beta)\geq F(\sigma,\tau_\beta) $ as the sufficient condition.
\end{proof}

\section{Technical Lemmas}\label{app:D_TL}
Here, we present a few technical tools that were used in the proofs of Theorem \ref{thm:flattest_exists}, in order to establish the construction of the flattest state $ \flattest $. These tools involve the functions
\begin{equation}\label{key}
F(m) = \sum_{i=1}^m p_i - p_{m+1} e^{\beta E_{m+1}}\sum_{i=1}^{m} e^{-\beta E_i}, \quad m\in\lbrace 1,d-1\rbrace,
\end{equation}
and 
\begin{equation}\label{key}
G(m) = p_{m-1} e^{\beta E_{m-1}} \sum_{i=m}^d e^{-\beta E_i} - \sum_{i=m}^d p_i,
\end{equation}
defined for $ \beta $-ordered eigenvalues of a block-diagonal state $ \rho $, denoted as $ \lbrace p_i\rbrace_i $.

Before starting, since we need to compare the value of these functions to the trace distance between $ \rho $ and $ \tau_\beta $, let us first rewrite $ \delta(\rho,\tau_\beta) $ into a more convenient expression. We have already seen that 
\begin{equation}\label{eq:td_pbetaordered}
\delta(\rho,\tau_\beta) = \sum_{i: p_i \geq \tau_i} p_i-\tau_i.
\end{equation}
We know that $ \lbrace p_i\rbrace_i $ has been $ \beta $-ordered. Moreover, we also know that for the constant $ Z_\beta $, $ p_i e^{\beta E_i} \geq \frac{1}{Z_\beta} $
is equivalent to $ p_i \geq \tau_i $. Therefore, the summation in Eq.~\eqref{eq:td_pbetaordered} may be simplified: there exists some integer $ 1\leq k\leq d-1 $ such that 
\begin{equation}\label{eq:k_td}
\delta(\rho,\tau_\beta) = \sum_{i=1}^k p_i-\tau_i = \sum_{i=1}^k p_i - \frac{1}{Z_\beta}\sum_{i=1}^k e^{-\beta E_i}.
\end{equation}
With this knowledge, we may proceed to prove certain properties of $ F(m) $ and $ G(m) $ in the subsequent lemmas.
\begin{Lemma}\label{lem:Fm_properties}
	The function $ F(m) $ is non-decreasing with respect to $ m $, while $ F(1)\geq 0 $. Moreover,  let $ 1\leq k\leq d-1 $ such that 
	\begin{equation}\label{eq:D5}
	\delta(\rho,\tau_\beta) = \sum_{i=1}^k p_i-\tau_i = \sum_{i=1}^k p_i - \frac{1}{Z_\beta}\sum_{i=1}^k e^{-\beta E_i}.
	\end{equation}
	Then we have $ F(k) \geq \delta(\rho,\tau_\beta) $. This also automatically implies that $ F(d-1) \geq \delta(\rho,\tau_\beta) $.
\end{Lemma}
\begin{proof}
	It is straightforward to see that since $ p_1 e^{\beta E_1} \geq p_2 e^{\beta E_2} $, we have
	$ 	F(1) = p_1 - p_2 e^{\beta E_2}e^{-\beta E_1} \geq 0 $. On the other hand,
	\begin{align*}
	F(m+1) &= \sum_{i=1}^{m+1} p_i - p_{m+2} e^{\beta E_{m+2}}\cdot\sum_{i=1}^{m+1} e^{-\beta E_i} \\
	&= \sum_{i=1}^{m} p_i + p_{m+1} - p_{m+2} e^{\beta E_{m+2}}\cdot\sum_{i=1}^{m} e^{-\beta E_i} \\
	&\quad- p_{m+2} e^{\beta E_{m+2}}e^{-\beta E_{m+1}}\\
	&\geq \sum_{i=1}^{m} p_i- p_{m+2} e^{\beta E_{m+2}}\cdot\sum_{i=1}^{m} e^{-\beta E_i} = F(m).
	\end{align*}
	The first equality simply comes from extracting out the $ (m+1) $-index from both summations, and the inequality comes from noting that the eigenvalues are $ \beta $-ordered, namely for any $ m $, we have
	$ p_{m+1} e^{\beta E_{m+1}} \geq p_{m+2} e^{\beta E_{m+2}}  $.
	
	The last item to prove is that for the integer $ k $ that gives rise to Eq.~\eqref{eq:D5}, we have $ F(k) \geq \delta(\rho,\tau_\beta) $. To do so, let us expand:
	\begin{align*}
	F(k) &= \sum_{i=1}^{k} p_i - p_{k+1} e^{\beta E_{k+1}}\cdot\sum_{i=1}^{k} e^{-\beta E_i}\\ &\geq  \sum_{i=1}^{k} p_i - \frac{1}{Z_\beta}\cdot\sum_{i=1}^{k} e^{-\beta E_i} = \delta(\rho,\tau_\beta).
	\end{align*}
\end{proof}
\begin{Lemma}\label{lem:Gm_properties}
	The function $ G(m) $ is non-increasing in $ m \in\lbrace 2,d\rbrace $, and $ G(d) \geq 0 $. Moreover, let $ 1\leq k\leq d-1 $ such that
	\begin{equation}
	\delta(\rho,\tau_\beta) = \sum_{i=1}^k p_i-\tau_i = \sum_{i=1}^k p_i - \frac{1}{Z_\beta}\sum_{i=1}^k e^{-\beta E_i},
	\end{equation}
	where $ \delta(\rho,\tau_\beta) $ is the trace distance between $ \rho $ and the thermal state $ \tau_\beta $.
	Then we have $ G(k) \geq \delta(\rho,\tau_\beta) $. This also automatically implies that $ G(2) \geq \delta(\rho,\tau_\beta) $.
\end{Lemma}
\begin{proof}
	The proof is rather similar to Lemma \ref{lem:Fm_properties}. First of all, by evaluating $ G(d) $, we have
	\begin{equation}\label{key}
	G(d) = p_{d-1} e^{\beta E_{d-1}} e^{-\beta E_{d}} - p_{d} \geq 0
	\end{equation}
	since by $ \beta $-ordering, $ p_{d-1} e^{\beta E_{d-1}}\geq p_{d} e^{\beta E_{d}} $.
	Subsequently, we have that 
	\begin{align*}
	&G(m) = p_{m-1} e^{\beta E_{m-1}}\cdot\sum_{i=m}^{d} e^{-\beta E_i}  - \sum_{i=m}^{d} p_i \\
	&= p_{m-1} e^{\beta E_{m-1}} \cdot\sum_{i=m+1}^{d} e^{-\beta E_i} + p_{m-1} e^{\beta E_{m-1}}e^{-\beta E_{m}} \\
	&\quad -\sum_{i=m+1}^d p_i - p_m\\
	&\geq p_{m} e^{\beta E_{m}} \cdot\sum_{i=m+1}^{d} e^{-\beta E_i} + p_{m-1} e^{\beta E_{m-1}}e^{-\beta E_{m}} \\
	& \quad-\sum_{i=m+1}^d p_i - p_m\\
	&= p_{m} e^{\beta E_{m}} \cdot\sum_{i=m+1}^{d} e^{-\beta E_i} \\
	&\quad -\sum_{i=m+1}^d p_i + p_{m-1} e^{\beta E_{m-1}}e^{-\beta E_{m}} - p_m\\&\geq G(m+1).
	\end{align*}
	To compare $ G(k) $ with $ \delta(\rho,\tau_\beta) $, let us rewrite Eq.~\eqref{eq:k_td}:
	\begin{eqnarray}\label{eq:compare_tr_wt_G2}
	\delta(\rho,\tau_\beta) 
	&=& 1- \sum_{i=k+1}^d p_i - \frac{1}{Z_\beta}\cdot \left(Z_\beta - \sum_{i=k+1}^d e^{-\beta E_i}\right) \nonumber\\
	&=& \frac{1}{Z_\beta}\cdot\sum_{i=k+1}^d e^{-\beta E_i} - \sum_{i=k+1}^d p_i.
	\end{eqnarray}
	Subsequently, by evaluating 
	\begin{align*}
	G(k+1) &= p_{k} e^{\beta E_{k}} \sum_{i=k+1}^d e^{-\beta E_i} -  \sum_{i=k+1}^d p_i \\
	&\geq \frac{1}{Z_\beta}\sum_{i=k+1}^d e^{-\beta E_i} - \sum_{i=k+1}^d p_i = \delta(\rho,\tau_\beta).
	\end{align*}
\end{proof}

By combining the properties of $ F(m) $ and $ G(m) $ proven in Lemma \ref{lem:Fm_properties} and \ref{lem:Gm_properties}, we can then make a statement about how $ N $ and $ M $ as chosen in the proof of Theorem \ref{thm:flattest_exists} relates, namely when $ \varepsilon < \delta(\rho,\tau_\beta) $, it is always true that $ M\leq N $.
\begin{Lemma}\label{lem6}
	For any value of $ \varepsilon $ between the interval $ 0\leq \varepsilon < \delta(\rho,\tau_\beta)$, consider the smallest integer $ 1 \leq M < d-1 $ where $ \epsilon \leq F(M) $. Furthermore, let $ 2 < N < d $ be the largest integer such that $ \varepsilon \leq G(N) $.
	Then $ M \leq N $.
\end{Lemma}
\begin{proof}
	By Lemma \ref{lem:Fm_properties} and \ref{lem:Gm_properties}, we know that there exists an integer $ 1\leq k\leq d-1 $ such that $ F(k) \geq \delta(\rho,\tau_\beta) > \varepsilon $, and also $ G(k+1) \geq \delta(\rho,\tau_\beta) > \varepsilon $.
	By Lemma \ref{lem:Fm_properties}, since $ F(m) $ is non-decreasing in $ m $, and since $ M $ is the smallest integer such that $ F(M) \geq \varepsilon $, this implies that $ M\leq k $ has to be true. On the other hand, by Lemma \ref{lem:Gm_properties} we know that $ G(m) $ is non-increasing in $ m $. Since $ N $ is the largest integer such that $ G(k) \geq\varepsilon $, then we know $ N\geq k+1 $. This implies that $ M\leq N $.
\end{proof}
\end{document}